\documentclass[11pt]{article}
\usepackage[T1]{fontenc}
\usepackage[letterpaper, margin=1in]{geometry}
\usepackage{graphicx} %
\usepackage{amsmath}
\usepackage{amssymb}
\usepackage{amsthm}
\usepackage{mathabx,mathtools}
\usepackage{bbm}
\usepackage[colorlinks=true, citecolor=blue, urlcolor=blue, pagebackref=true]{hyperref}

\usepackage{xcolor}
\usepackage{stmaryrd}
\usepackage{subcaption}
\usepackage{paralist}
\usepackage{booktabs}
\usepackage{enumitem}
\usepackage{undertilde}
\usepackage{cleveref}

\newtheorem{theorem}{Theorem}[section]
\newtheorem{lemma}{Lemma}[section]

\newtheorem{corollary}{Corollary}[section]
\newtheorem{proposition}{Proposition}[section]

\theoremstyle{definition}
\newtheorem{definition}{Definition}[section]
\newtheorem{problem}{Problem}[section]
\newtheorem{remark}{Remark}[section]
\newtheorem{conjecture}{Conjecture}[section]

\newcommand{\ov}{\overline}

\DeclareMathOperator*{\bbE}{\mathbb{E}}
\newcommand{\bbF}{\mathbb{F}}

\newcommand{\bbR}{\mathbb{R}}
\newcommand{\bbZ}{\mathbb{Z}}

\newcommand{\bbN}{\mathbb{N}}

\newcommand{\calR}{\mathcal{R}}
\newcommand{\calX}{\mathcal{X}}

\newcommand{\degen}{\trianglelefteq}

\newcommand{\ang}[1]{\langle #1 \rangle}
\newcommand{\iv}[1]{\llbracket #1 \rrbracket}

\newcommand{\Sym}{\mathfrak{S}}
\newcommand{\specMM}{\Delta_{\mathrm{MM}}}
\newcommand{\RMM}{\calR_{\mathrm{MM}}}

\newcommand{\axref}[2]{\hyperref[ax:#1]{#2}}
\newcommand{\axitem}[2]{\item[\textbf{#2}]\phantomsection\label{ax:#1}}

\newcommand{\cw}{\mathrm{cw}}
\newcommand{\CW}{\mathrm{CW}}

\newcommand{\sfT}{\mathsf{T}}

\renewcommand{\setminus}{\smallsetminus}

\newcommand{\dif}{\mathop{}\!\mathrm{d}}

\DeclareMathOperator{\id}{id}

\DeclareMathOperator{\Span}{span}

\DeclareMathOperator{\Rk}{R}
\DeclareMathOperator{\BR}{\underline{R}}
\DeclareMathOperator{\AR}{\utilde{\mathrm{R}}}
\DeclareMathOperator{\AQ}{\utilde{\mathrm{Q}}}
\DeclareMathOperator{\Q}{Q}
\DeclareMathOperator{\BQ}{\underline Q}

\newcommand{\email}[1]{\href{mailto:#1}{\texttt{#1}}}

\setlist[enumerate, 1]{label = {(\arabic*)}}

\numberwithin{equation}{section}

\newcommand{\myparagraph}[1]{\paragraph{#1}}

\title{Asymptotic Rank Speedup Theorems, Revisited}%
\author{Josh Alman\thanks{Columbia University. \email{josh@cs.columbia.edu}. Supported in part by NSF Grant CCF-2238221 and a Packard Foundation Fellowship.} \and
    Baitian Li\thanks{Columbia University. \email{bl3052@columbia.edu}. Supported in part by NSF Grant CCF-2238221, a Packard Foundation Fellowship, and a Columbia SEAS Presidential Fellowship.}}
\date{}

\begin{document}

\maketitle

\begin{abstract}
    Motivated by fast matrix multiplication and recent connections between asymptotic tensor rank and fine-grained complexity, we revisit  classical tools from the matrix multiplication literature and develop a framework for obtaining improved asymptotic rank upper bounds for tensors beyond matrix multiplication.

    In the 1980s, Coppersmith--Winograd and Strassen discovered a series of \emph{speedup theorems} for asymptotic rank: in certain regimes, one can extract additional terms from a border rank upper bound on a tensor $T$, and then use these terms to obtain an improved asymptotic rank of $T$. We establish general speedup theorems that subsume these results and enable quantitative improvements. Two representative applications are:
    \begin{enumerate}
        \item The asymptotic rank of the small Coppersmith--Winograd tensor $\cw_q$ is less than its border rank. For instance, we prove $\AR(\cw_2) < 3.931$, improving on $\BR(\cw_2)=4$. It is known that $\AR(\cw_2)=3$ would imply $\omega=2$.
        \item A general improvement over Strassen's bound: we obtain an upper bound below $d^{2\omega/3}$ on the asymptotic rank of any $d\times d\times d$ tensor.
    \end{enumerate}
    To make full use of speedups, we analyze degenerations in which both sides are nontrivial direct sums, a setting where the optimal quantitative bound one can achieve was previously unclear. We do so via an approach we call \emph{Strassen calculus}: a systematic method for converting such degeneration data into explicit asymptotic rank bounds using Strassen's theory of the asymptotic spectrum.
\end{abstract}

\section{Introduction}

The exponent of matrix multiplication, $\omega$, is defined as the smallest real number such that two $n \times n$ matrices can be multiplied using at most $n^{\omega + o(1)}$ field operations. Because of the fundamental and widespread applications of matrix multiplication algorithms, determining the value of $\omega$ has become one of the most important open problems in computer science. Strassen's famous algorithm~\cite{strassen1969gaussian} showed the first nontrivial bound $\omega \leq \log_2(7) < 2.81$. Since then, a long line of follow-up work has culminated in the current best bound $\omega < 2.371339$~\cite{ADWXXZ25MoreAsym}. The only known lower bound is the straightforward $\omega \geq 2$, since the input and output sizes are $n^2$. In fact, it is frequently conjectured that $\omega = 2$.

It has been known since Strassen's work~\cite{strassen1973vermeidung} that determining $\omega$ is equivalent to understanding the \emph{asymptotic rank} of the matrix multiplication tensor. In fact, there are three different rank notions for a (3-mode/3-dimensional) tensor $T$ which have played a prominent role in the study of matrix multiplication algorithms: the rank $\Rk(T)$, the \emph{border rank} $\BR(T)$ (roughly, the smallest $r$ such that $T$ is the limit of a sequence of tensors of rank $r$), and the asymptotic rank $\AR(T)$, where one takes the ranks of large tensor powers of $T$: 
\[ \AR(T) := \lim_{n \to \infty} (\Rk(T^{\otimes n})^{1/n}). \]

For instance, Strassen's algorithm~\cite{strassen1969gaussian} showed that the rank of the tensor for multiplying $2 \times 2$ matrices is at most 7. More generally, because of the recursive structure of matrix multiplication under tensor products, to upper bound $\omega$ it suffices to bound the asymptotic rank of matrix multiplication tensors. It also suffices to bound the border rank because of a connection due to Bini~\cite{bini1980relations} which showed that, for any tensor $T$, we have
\[ \AR(T) \leq \BR(T) \leq \Rk(T). \]

More recently, the asymptotic rank of tensors has found other algorithmic applications beyond matrix multiplication. Bj{\"o}rklund and Kaski~\cite{bjorklund2024asymptotic} and Pratt~\cite{Pratt24SCC} showed that improved asymptotic rank bounds for certain tensors would give a faster algorithm for the \emph{Set Cover} problem, which is conjectured in fine-grained complexity theory to be impossible. Follow-up work showed similar connections for other algorithmic problems including computing the chromatic number of a graph~\cite{bjorklund2025fast} and computing the permanent of a matrix~\cite{bjorklund2025kronecker}.
Many exponential and parameterized algorithms rely on a (generalized) convolution problem, the complexity of which is determined by the asymptotic rank of the corresponding tensor \cite{BCLP26Conv}.

Amazingly, Strassen conjectured that \emph{every} tensor has the minimum possible asymptotic rank given its dimensions (as long as it is ``tight'', a technical condition satisfied by most tensors of interest):

\begin{conjecture}[Asymptotic Rank Conjecture~\cite{strassen1994algebra}]
Every $d \times d \times d$ tensor $T$ which is concise and tight has $\AR(T) = d$.
\end{conjecture}
This would imply $\omega=2$ and lead to breakthrough algorithms for all the aforementioned problems.

Unfortunately, despite the importance of asymptotic rank, we don't know many techniques for bounding the asymptotic rank of a tensor $T$ beyond bounding the border rank of $T$ or a small tensor power of $T$. Even this approach, which often gives suboptimal bounds, is difficult to use, as bounding the rank or border rank of fairly small tensors has proven to be quite challenging. For example, it remains an open problem to determine the rank or border rank of $\langle 3,3,3 \rangle$, the tensor for multiplying $3 \times 3$ matrices (which is a $9 \times 9 \times 9$ tensor).

Beyond this, the two most successful techniques for bounding asymptotic rank for the purposes of bounding $\omega$ have been the laser method (a combinatorial technique for restricting tensor powers of structured tensors into direct sums of matrix multiplication tensors), and the \emph{asymptotic sum inequality} (a.k.a.~\emph{Sch\"onhage's $\tau$ theorem}) for taking advantage of such direct sums:

\begin{theorem}[Asymptotic Sum Inequality~\cite{Sch81Tau}]
    Suppose a direct sum of matrix multiplication tensors has a border rank bound
\[ \BR \left( \bigoplus_{i=1}^\ell \ang{n_i, m_i, p_i} \right) \leq r, \]
and that this is nontrivial (i.e., $n_im_ip_i > 1$ for at least one $i$). Let $\tau$ be the unique solution of the transcendental equation
\[ \sum_{i=1}^\ell (n_i m_i p_i)^\tau = r. \]
Then, $\omega \leq 3\tau$.
\end{theorem}

It was then shown in follow-up work of Coppersmith and Winograd \cite{CW82Improve} that the asymptotic sum inequality is \emph{never tight}:
One always has the strict inequality $\omega < 3\tau$! Their proof is constructive, in the sense that they found a systematic way to
convert a given direct sum construction $\BR(T)\leq r$ into another one $\BR(T') \leq r'$, which yields a \emph{more effective} bound on the matrix multiplication exponent $\omega$.

This type of result is reminiscent of Blum's classical \emph{speedup} theorem for Turing machines \cite{Blum67Speedup}, and so we refer to such results as speedup theorems in this paper. That said, Blum's speedup theorem concerns non-constructible time functions, whereas the tensor speedup theorems we discuss here concern explicit problems of great interest like matrix multiplication.

\myparagraph{Our approach.}

In this paper, we revisit the approach of bounding the asymptotic ranks of tensors via speedup theorems. We establish a generalization of prior speedup theorems, and analyze them using a method we call \emph{Strassen calculus} which makes use of Strassen's theory of the asymptotic spectrum, in order to give new quantitative asymptotic rank bounds. We will discuss this in much more detail in Section~\ref{sec:techniqueoverview} below, but we first present some applications of our approach.

\myparagraph{Small Coppersmith--Winograd tensor.}

As a main application, we give the first nontrivial upper bound on the asymptotic rank of the \emph{small Coppersmith--Winograd tensor}, $\cw_q$. This is, for any integer $q \geq 2$, the $(q+1) \times (q+1) \times (q+1)$ tensor given by
\[ \cw_q = \sum_{i=1}^q x_0 y_i z_i + x_i y_0 z_i + x_i y_i z_0. \] 

The $\cw_q$ tensor was introduced in the celebrated work of
Coppersmith and Winograd~\cite{CW90Laser}. They further developed Strassen's laser method, and applied it first to $\cw_q$ as their easy (warmup) construction \cite[\S 6]{CW90Laser}, then later to the \emph{big Coppersmith--Winograd tensor}, $\CW_q$, the $(q+2) \times (q+2) \times (q+2)$ tensor which adds three ``corner terms'' to $\cw_q$,
\[ \CW_q = x_0 y_0 z_{q+1} + x_0 y_{q+1} z_0 + x_{q+1} y_0 z_0 + \cw_q. \] 

Coppersmith and Winograd showed that both $\cw_q$ and $\CW_q$ have border rank $q+2$ (and thus, asymptotic rank at most $q+2$). By virtue of the extra corner terms, they found $\CW_q$ to be more valuable in their laser method analysis, giving their bound $\omega < 2.376$. The big tensor $\CW_q$ has since been used for all record-holding bounds on $\omega$, including the current bound $\omega < 2.371339$~\cite{ADWXXZ25MoreAsym}.

However, the tensor $\CW_q$ is now known to be subject to a number of \emph{barrier} results~\cite{ambainis2015fast,AW21limits,CVZ21Barrier,alman2021limits} which show that it is \emph{not} valuable enough to prove $\omega=2$. (The exact numerical barrier depends on which techniques one allows to analyze $\CW_q$, but the most general shows that even with the theoretically optimal asymptotic restrictions/degenerations, one cannot achieve a better bound on $\omega$ than $2.16$ using $\CW_q$.) 

There is no known barrier like this to using $\cw_q$. The key difference is in the dimensions of $\cw_q$ compared to $\CW_q$. Since $\CW_q$ is a concise (roughly meaning ``nondegenerate'') tensor of dimensions $(q+2) \times (q+2) \times (q+2)$, it follows that $\AR(\CW_q) \geq q+2$, and so the bound $\AR(\CW_q) \leq q+2$ from~\cite{CW90Laser} is tight. By contrast, for $\cw_q$ there is a gap between the lower bound $\AR(\cw_q) \geq q+1$ from its dimensions, and the upper bound $\AR(\CW_q) \leq q+2$ from~\cite{CW90Laser}. Using their laser method analysis, Coppersmith and Winograd~\cite[\S 11]{CW90Laser} observed that closing this gap could have major implications:

\begin{theorem}[\cite{CW90Laser}]\label{intro:omega2fromcw2}
    If $\AR(\cw_2) = 3$ then $\omega = 2$.
\end{theorem}

Prior to our work, no bound better than Coppersmith and Winograd's border rank bound $\AR(\cw_2) \leq \BR(\cw_2) = 4$ was known. But because of Theorem~\ref{intro:omega2fromcw2}, many prior works have studied and asked about improved bounds on $\AR(\cw_q)$. See, for instance,~\cite[Remark~15.44 and Exercise~15.25]{BCS13ACT}, \cite[Problem~9.8]{Blaser13FMM}, and \cite[Question~3]{ConnerLGV20ITCS}. 

Most recently, Conner, Landsberg, Gesmundo, and Ventura~\cite{ConnerLGV20ITCS} studied whether the border ranks of small Kronecker powers $\cw_q^{\otimes k}$ for $k=2, 3$ could be used to improve the upper bound on $\AR(\cw_q)$, but they achieved a negative result: $\BR(\cw_q^{\otimes 2}) = (q+2)^2$ (for all $q>2$) and $\BR(\cw_q^{\otimes 3}) = (q+2)^3$ (for all $q>4$). This suggests that one may need to take large Kronecker powers of $\cw_q$ before any border rank submultiplicativity is observed.

As an application of our approach for bounding the asymptotic ranks of tensors, we achieve the first upper bound improvement:

\begin{theorem}\label{thm:introcw2}
    $\AR(\cw_2) < 3.931$.
\end{theorem}
More generally, for any $q \geq 2$, we show there is a $\delta_q > 0$ such that $\AR(\cw_q) < q + 2 - \delta_q$; a table of these bounds for $2 \leq q \leq 10$ is given in Table~\ref{tab:small_cw} in Section~\ref{subsec:smallcw} below.

We remark that, in hindsight, the existence of such a $\delta_q>0$ actually follows from Strassen's original work on the asymptotic spectrum \cite[Theorem~3.14]{Str88AsymSpec}. That said, we believe that a quantitative bound on how large $\delta_q$ is requires the new analysis we introduce here. We will discuss this more in the technique overview in the next section.

\myparagraph{Improved general upper bound on asymptotic rank.}
Beyond Theorem~\ref{thm:introcw2}, we prove additional upper bounds on the asymptotic rank of tensors, as well as new restrictions involving direct sums of tensors of interest. As one example, we obtain a new upper bound on the maximum possible asymptotic rank of a $d\times d\times d$ tensor.

For small $d$, the generic rank bounds \cite{Str83GenericR,Lic85GenericR} already yield nontrivial estimates, and recent geometric methods can slightly improve them \cite{KM25ARC,Lee26AR}; however, these bounds remain $\Omega(d^2)$. For larger $d$, the best known general upper bound instead comes from a generic reduction of an arbitrary tensor to matrix multiplication, which implies  that every $d\times d\times d$ tensor $T$ satisfies $\AR(T)\le d^{\frac23\omega}$ \cite{Str88AsymSpec}.

It is natural to wonder whether this can be improved, and recently Kaski~\cite[Open~Problem~1]{Kaski2025Talk} asked this explicitly in the context of further understanding sequences of tensors related to Strassen's asymptotic rank conjecture. We improve it as follows:

\begin{theorem} \label{thm:intro_general_oneslice}
    For any $d \geq 3$, any $d\times d\times d$ tensor $T$ satisfies
    \[ \AR(T) \leq \sqrt{d^{\frac 4 3 \omega} -
    \left\lfloor \frac d 3 \right\rfloor\cdot (2^{\frac \omega 3} - 1)d^{\frac 2 3 \omega}} < d^{\frac 2 3 \omega}.\]
\end{theorem}

\section{Technique overview} \label{sec:techniqueoverview}

\subsection{Strassen calculus: Bounds via asymptotic spectrum} \label{subsec:strassencalculusintro}

Traditionally, techniques for bounding the matrix multiplication exponent $\omega$ proceed by constructing a chain of border-rank upper bounds. This requires carefully tracking the orders of magnitude of intermediate tensors, and typically culminates in a degeneration of the form $\ang{n,n,n}\degen \ang{r}$, which bounds $\omega$ via the basic inequality $\omega \leq \log_n \BR(\ang{n,n,n})$. A major drawback is that these arguments often rely on \emph{ad hoc} derivation chains. A prominent example is the standard proof of the asymptotic sum inequality (cf.~\cite[\S 15.5]{BCS13ACT}, \cite[\S 7]{Blaser13FMM}), which hinges on a specific and rather involved degeneration.

Furthermore, it is often unclear whether such a derivation yields the optimal
bound theoretically obtainable from a given set of inequalities. For instance, in 1981, 
Sch\"onhage \cite{Sch81Tau} obtained the (record-holding at the time) bound $\omega < 2.521801$ by proving and analyzing the direct
sum degeneration (specifically, setting $k=5$ and $n=10$):
\[ \ang{1,k,2n} \oplus \ang{n,2,k} \oplus \ang{2k,n,1} \degen \ang{2(n+1)(k+1)}
\oplus (k+1) \odot \ang{1,1,2}. \]
Sch\"onhage's analysis, however, was not optimal. Romani \cite{Rom82DSum} later
improved this estimate to $\omega < 2.516649$ by exploiting the same identity more
effectively. Even this is not tight: the best bound obtainable from the identity is
$\omega \leq 3\tau < 2.516094$, where $\tau$ is the solution of the following equation (which looks
like a natural generalization of asymptotic sum inequality):
\[ 3\cdot (2kn)^\tau = 2(n+1)(k+1) + (k+1) \cdot 2^\tau,
\quad k=5, n=10. \]
This was simultaneously noted
by Pan \cite[Appendix A]{Pan81} and Sch\"onhage
(and as claimed in Pan's paper, a full proof is included in a
unpublished manuscript of Sch\"onhage). As far as we are aware,
there is still no accessible proof of this analysis, likely because the community's attention shifted to other methods to design matrix multiplication algorithms. 
We provide a proof in \Cref{sec:omega} below that 
$\omega \leq 2.516094$ can be achieved by analyzing this direct sum identity, and is the \emph{optimal} bound one can achieve from it.

In \Cref{sec:str_calc}, we introduce \emph{Strassen calculus}---a systematic framework for deriving bounds from a fixed collection of degenerations between direct sums of tensors---via Strassen's theory of the \emph{asymptotic spectrum} \cite{Str88AsymSpec,Str91SuppFun,Str05Kom}.

Informally, the asymptotic spectrum provides a family of numerical ``measures'' for tensors. Concretely, an element of the asymptotic spectrum $\calX$ is a function
\[
    \phi:\{\text{$\bbF$-tensors}\}\to \bbR_{\ge 0}
\]
that behaves well with respect to the tensor operations. %
We refer to such a function $\phi$ as a \emph{spectral point}; see \Cref{sec:str_calc} for the precise definitions.

The key structural statement underlying the theory of the asymptotic spectrum is Strassen duality (\Cref{thm:duality}). Roughly speaking, it says that asymptotic degeneration relations are exactly the relations that hold after applying \emph{all} spectral points: if $T$ asymptotically degenerates to $S$, then $\phi(T)\le \phi(S)$ for every $\phi$, and conversely, if this inequality holds for all $\phi$, then this implies the existence of the asymptotic degeneration. In particular, this characterizes asymptotic rank via a max formula,
\[
    \AR(T)=\max_{\phi\in\calX}\phi(T).
\]
Conceptually, this turns questions about asymptotic rank into optimization problems over the asymptotic spectrum.

Of course, the space $\calX$ of all spectral points is defined \emph{nonconstructively} and may be \emph{infinite-dimensional} \emph{a priori}, so constructing explicit spectral points is generally difficult, although there has been significant progress in the past decade \cite{CVZ23QuantumFun,CLZ23WSR,SDW26Fun}. Our approach is therefore not to ``parametrize'' $\calX$ directly. Instead, we start from a fixed collection of tensor degenerations, translate them into numerical inequalities that must be satisfied by every spectral point, and then use Strassen duality to turn the resulting constraints into explicit upper bounds on asymptotic rank.
In favorable regimes one can solve the corresponding optimization problems analytically; more generally, one can obtain sharp bounds via computer-aided evaluation.
Moreover, Strassen calculus is \emph{complete}: for a fixed set of degenerations, optimizing over these constraints yields the best bound obtainable from those relations.
In \Cref{sec:omega}, we provide a rigorous treatment of the notion of ``optimal usage'',
and show the completeness under this formalism.

We emphasize that although the theory of the asymptotic spectrum was originally motivated by the study of tensor rank, its subsequent development has followed a curious trajectory. While the framework has flourished in other areas such as (quantum) information theory, its application to the \emph{asymptotic rank} of tensors has remained largely conceptual. Prior to this work, the theory occupied a somewhat dormant role in its own founding domain, as to the best of our knowledge, \emph{no} practical instances were identified to yield concrete, improved upper bounds.

The following examples illustrate the utility of the framework in other areas, and highlight the gap between its success elsewhere and its (so far) limited impact on tensor rank:
\begin{itemize}
    \item \emph{Strassen's Foundation.} In his seminal paper \cite{Str88AsymSpec}, Strassen introduced the asymptotic spectrum to provide a new proof of the asymptotic sum inequality, setting the stage for a theory that, while foundational, remained difficult to apply to specific new bounds.
    \item \emph{Barriers to Fast Matrix Multiplication.} Improvements to the exponent $\omega$ over the past thirty years have largely relied on analyzing the big Coppersmith--Winograd tensor $\CW_q$. However, Christandl, Vrana, and Zuiddam \cite{CVZ21Barrier} utilized spectral points to show that the laser method---even in a highly generalized form---applied on big CW tensors, cannot yield an upper bound better than $2.16$. Subsequent work \cite{CLLZ25Barrier} has established analogous barriers for the rectangular case.
    \item \emph{Graph Theory.} Zuiddam \cite{Zui19Graph} applied the general theory of asymptotic spectrum (not confined to tensors) to study the Shannon capacity $\Theta(G)$ of graphs. Wigderson and Zuiddam \cite[Theorem 4.1]{WZ25Spectra} utilized the asymptotic spectrum to show a striking equivalence:
    for graphs $G$ and $H$, where $G\boxtimes H$ and $G\sqcup H$ denote the strong product and disjoint union respectively,
    $\Theta(G \boxtimes H) > \Theta(G) \Theta(H)$ if and only if $\Theta(G \sqcup H) > \Theta(G) + \Theta(H)$.\footnote{To appreciate this result, note that examples of graphs $G, H$ satisfying the first strict inequality were found in \cite{Hae79Shannon}, yet it took decades before examples satisfying the second were discovered \cite{Alon98Shannon}.}
    \item \emph{Circuit Complexity.} Robere and Zuiddam \cite{RZ21CatB} developed an analogous theory to study
    and generalize some results in the complexity of catalytic branching programs.
    Alman and Li \cite{AL25Depth2} applied asymptotic spectrum theory to derive a duality theorem for depth-2 circuits, resolving the problem of optimally combining fixed constructions in that context.
\end{itemize}

These developments demonstrate that Strassen duality is a mature and productive framework in other contexts. Our work aims to \emph{re-apply} this machinery to its original setting by identifying concrete tensor degenerations for which it provably yields new, improved upper bounds.

In \Cref{sec:str_calc} we provide two self-contained examples of Strassen calculus. They focus on how one can turn degenerations of the form
\[
    T \oplus \ang{1,t,1} \degen \ang{r} \oplus \ang{1,s,1}
\]
into explicit upper bounds on asymptotic rank. On the one hand, from such identities---with rectangular matrix multiplication tensors on the right-hand side---it was not previously understood what the optimal resulting asymptotic rank bound is. On the other hand, relations of this type (and more complicated variants) arise naturally in our generalized speedup theorems. For this reason, it is important to use the Strassen calculus in conjunction with our speedups.

In particular, in \Cref{subsec:toy}, our first example is a \emph{toy example} illustrating the basic principles of the framework, together with a detailed derivation. Despite its simplicity, it already yields a novel and interesting consequence: every $3\times 3\times 3$ tensor $T$ satisfies $\AR(T)\leq 3+2^{\omega/3}$ (\Cref{cor:dim_three_ar}), giving a first glimpse of the calculus in action. (Previously, to the best of our knowledge, the upper bound on $\AR(T)$ for a generic $3 \times 3 \times 3$ tensor $T$ was $5$, which is the generic rank. This result improves it to $3 + 2^{\omega/3} < 4.7297$.)

\subsection{Speedup theorems}

In \Cref{sec:speedup}, we prove a general speedup theorem that subsumes the results of Coppersmith--Winograd \cite{CW82Improve} and Strassen \cite[Theorem 3.14]{Str88AsymSpec}. The Coppersmith--Winograd theorem fits the following paradigm: starting from a border rank bound $T \degen \ang{r}$, one can (in certain regimes) ``speed up'' the degeneration by taking a direct sum with a one-slice matrix multiplication tensor on both sides, obtaining
\[
    T \oplus \ang{1,t,1} \degen \ang{r} \oplus \ang{1,s,1}.
\]

Strassen later generalized this by showing that the right-hand side does not necessarily need to be a diagonal tensor (i.e., a border rank bound). Specifically, a degeneration $T \degen S$ can also be accelerated to $T \oplus \ang{1,t,1} \degen S \oplus \ang{1,s,1}$.
Our main theorem further generalizes this idea: under suitable hypotheses, it allows one to add \emph{arbitrary} tensors on both sides of a degeneration, producing relations of the form $T \oplus T' \degen S \oplus S'$.

In \Cref{sec:instantiation}, we present three concrete instantiations of the general speedup theorem. Each has different strengths and ranges of applicability; together, they form the main toolkit we use to derive explicit upper bounds on asymptotic rank.
\begin{itemize}
    \item The first result relates a certain variant of the min-rank problem to one-slice speedups of the form $T\oplus \ang{1,t,1} \degen S\oplus \ang{1,s,1}$, thereby recovering Strassen's speedup theorem.
    \item The second result shows that, in a suitable regime, a degeneration of the form $T \oplus \ang{1,t,1} \degen S'$ can be further accelerated to $T \oplus \ang{1,t',1} \degen S'$ for some $t' > t$, while keeping the right-hand side fixed. In particular, if one takes an identity obtained from the first result and takes a tensor power of both sides, this acceleration always applies; thus the first method can be \emph{iterated}. This generalizes a result of Coppersmith and Winograd \cite[Theorem B]{CW82Improve}.
    \item The third result shows that, in a certain regime, one can add multiple one-slice tensors on both sides, obtaining degenerations of the form $T\oplus \bigoplus_i \ang{1,t_i,1}\degen S \oplus \bigoplus_j \ang{1,s_j,1}$.
\end{itemize}
As a motivating illustration, we use these methods to show that any non-minimal border rank bound $\BR(T) \leq r$ can be strictly improved (a \emph{qualitative} fact first established by Strassen \cite[Theorem~3.14]{Str88AsymSpec}). The three instantiations above are formalized as \Cref{thm:app-oneslice}, \Cref{thm:iterate}, and \Cref{thm:speedup_via_compression}, and we give a \emph{quantitative} comparison of the improvement achieved by each method using Strassen calculus.

In \Cref{sec:appl}, we apply these methods to obtain explicit bounds for concrete tensors of interest.
\begin{itemize}
    \item In \Cref{subsec:smallcw}, we prove Theorem~\ref{thm:introcw2}, improving the asymptotic rank of the small Coppersmith--Winograd tensors $\cw_q$; for instance, we improve the best known bound for $\AR(\cw_2)$ from $4$ (coming from border rank) to $3.931$ (see \Cref{tab:small_cw}). %
    \item In \Cref{subsec:exponent}, we prove Theorem~\ref{thm:intro_general_oneslice}, improving general upper bounds on asymptotic rank. %
    \item In \Cref{subsec:directsum}, we use the speedup theorem to construct new degenerations between direct sums of matrix multiplication tensors in which the right-hand side is itself a direct sum of nontrivial matrix multiplication tensors. While this does not (at present) improve the matrix multiplication exponent, the resulting identities do not appear to follow easily from previously known ones, and may provide insights for future developments.
\end{itemize}

Finally, it is natural to wonder whether this method yields faster algorithms for problems such as matrix multiplication or set cover. We find that our techniques are still not effective enough, although we are optimistic that further improvements along these lines may succeed. See \Cref{rmk:cw_omega} and \Cref{rmk:reduce_mm} for further discussion.

\section{Preliminaries}

For a positive integer $n$, let $[n] := \{1,\dots,n\}$. We use Iverson bracket
notation: for a proposition $P$, let $\iv{P}=1$ if $P$ is true and $\iv{P}=0$
otherwise.

\subsection{Tensors}

Our notation largely follows \cite{DWZ23AsymHash,WXXZ24MatMul,ADWXXZ25MoreAsym}.

Fix a base field $\bbF$.
Let $U, V, W$ be finite-dimensional vector spaces with bases
$\{x_1,\dots,x_n\}$ of $U$, $\{y_1,\dots,y_m\}$ of $V$ and $\{z_1,\dots,z_p\}$ of $W$.
A tensor is an element $T\in U\otimes V\otimes W$. With respect to the chosen bases,
we write $T$ via its coefficients $a_{ijk} \in \bbF$ as
\[ T = \sum_{i=1}^n \sum_{j=1}^m \sum_{k=1}^p a_{ijk} x_i y_j z_k. \]
We say $U,V,W$ are the \emph{modes} of the tensor.

For two tensors $S\in U\otimes V\otimes W$ and $T\in U'\otimes V'\otimes W'$,
we write their direct sum $S\oplus T$ as the image of their sum inside the common space
$(U\oplus U') \otimes (V\oplus V') \otimes (W \oplus W')$,
and their tensor product as the natural image inside the tensor product space
$(U\otimes U') \otimes (V\otimes V') \otimes (W\otimes W')$.

A \emph{restriction} from $S\in U\otimes V\otimes W$ to $T \in U'\otimes V'\otimes W'$
is specified by linear maps $A\colon U \to U'$, $B\colon V \to V'$ and $C\colon W \to W'$
such that $T = (A\otimes B\otimes C)S$. Concretely,
\[ T = \sum_{i,j,k} a_{ijk} A(x_i) B(y_j) C(z_k). \]
In this case we say $S$ \emph{restricts to} $T$, and write $T\leq S$.

A \emph{degeneration} is a one-parameter family of linear maps
$A_\lambda, B_\lambda, C_\lambda$ (with entries rational in $\lambda$) such that
$(A_\lambda \otimes B_\lambda \otimes C_\lambda) S = T + O(\lambda^d)$ for some $d$. In other words, write each entry of $A_\lambda, B_\lambda, C_\lambda$ as a rational function in $\lambda$ and then expand the resulting expression as a Laurent series in $\lambda$:
\[ A_\lambda = \sum_{i\geq i_0} A_i \lambda^i, \]
and similarly for $B_\lambda, C_\lambda$, we have the $\lambda^\ell$-coefficient of $(A_\lambda \otimes B_\lambda \otimes C_\lambda) S$ is
\[ \sum_{i+j+k=\ell} (A_i \otimes B_j \otimes C_k) S. \]
We require that the $\lambda^0$-coefficient is exactly $T$, and that the $\lambda^\ell$-coefficient is zero for all $\ell < 0$. In this case we say $T$ is a degeneration of $S$, write $S$ \emph{degenerates to} $T$, and denote this by $T\degen S$.

Two central families of tensors are diagonal tensors and matrix multiplication tensors.
For any $n\geq 1$, the diagonal tensor is
\[ \ang{n} = \sum_{i=1}^n x_i y_i z_i, \]
and for any $n,m,p\geq 1$, the matrix multiplication tensor is
\[ \ang{n,m,p} = \sum_{i=1}^n \sum_{j=1}^m \sum_{k=1}^p x_{ij} y_{jk} z_{ki}. \]

The rank $\Rk(T)$ of a tensor $T$ is the minimum $r$ such that $T \leq \ang{r}$.
Similarly, the subrank $\Q(T)$ is the maximum $q$ such that $\ang{q} \leq T$ (equivalently, $T\geq \ang{q}$).
The \emph{border rank} $\BR(T)$ and \emph{border subrank} $\BQ(T)$ are defined analogously by replacing $\leq$ with $\degen$.

We write the $n$-fold direct sum of $T$ as $n \odot T$; note that $n \odot T$ is equivalent to $\ang{n} \otimes T$.

The \emph{asymptotic rank} and \emph{subrank} are defined as
\begin{align*}
    \AR(T) &= \lim_{n\to\infty} \Rk(T^{\otimes n})^{1/n},\\
    \AQ(T) &= \lim_{n\to\infty} \Q(T^{\otimes n})^{1/n}.
\end{align*}
By Fekete's lemma, for $\AR$ ($\AQ$ resp.), the $\lim_{n\to \infty}$ can also be replaced by $\inf_{n\geq 1}$ ($\sup_{n\geq 1}$ resp.).
Therefore, the matrix multiplication exponent $\omega$ is characterized by $\AR(\ang{n,n,n}) = n^\omega$.

\section{Strassen calculus} \label{sec:str_calc}

We first briefly recall the notion of the asymptotic spectrum. The asymptotic spectrum of $\bbF$-tensors
is the set $\calX = \calX_\bbF$ of functions $\phi$ from tensors to nonnegative real numbers $\bbR_{\geq 0}$ satisfying:
\begin{itemize}
    \item Additive under direct sum: $\phi(T \oplus S) = \phi(T) + \phi(S)$.
    \item Multiplicative under tensor product: $\phi(T\otimes S) = \phi(T) \phi(S)$.
    \item Normalized on diagonal tensors: $\phi(\ang{n}) = n$.
    \item Monotone under restriction (degeneration\footnote{It turns out that
    \emph{restriction} and \emph{degeneration} give equivalent definitions,
    i.e., they yield the same collection of functions; cf.~\cite[Theorem~3.1]{Str88AsymSpec}.}): $\phi(T) \leq \phi(S)$ when $T\leq S$ ($T\degen S$).
\end{itemize}

From this definition, it is immediate that $\calX$ only ``remembers'' asymptotic relations:
if $T^{\otimes n} \leq S^{\otimes n+o(n)}$, then $\phi(T) \leq \phi(S)$ for all $\phi\in\calX$.
Surprisingly, the converse holds as well; this is the \emph{Strassen duality theorem}
(cf.~\cite[Theorem~3.2]{Str88AsymSpec}, \cite[\S 3]{WZ25Spectra}).

\begin{theorem}[Strassen duality] \label{thm:duality}
    For two tensors $T, S$, the following are equivalent:
    \begin{enumerate}
        \item $T^{\otimes n} \leq S^{\otimes n + o(n)}$,
        \item $T^{\otimes n} \degen S^{\otimes n + o(n)}$,
        \item For any $\phi\in\calX$, have $\phi(T) \leq \phi(S)$.
    \end{enumerate}
    Moreover,
    \[ \AR(T) = \lim_{n\to\infty} \Rk(T^{\otimes n})^{1/n}
    = \lim_{n\to\infty} \BR(T^{\otimes n})^{1/n} = \max_{\phi\in\calX} \phi(T), \]
    and
    \[ \AQ(T) = \lim_{n\to\infty} \Q(T^{\otimes n})^{1/n} = 
    \lim_{n\to\infty} \BQ(T^{\otimes n})^{1/n} = \min_{\phi\in\calX} \phi(T). \]
\end{theorem}

For our purposes, it is enough to know Strassen duality; no further background on the asymptotic spectrum is required.
For a complete treatment---including further properties and applications---we refer to Strassen's founding paper \cite{Str88AsymSpec} and the survey of Wigderson and Zuiddam \cite{WZ25Spectra}.

\subsection{Matrix multiplication}

Before proceeding to the toy example, we first introduce a fundamental concept
to facilitate our use of Strassen calculus. 

\begin{proposition}[{\cite[Proposition~4.3]{Str88AsymSpec}}] \label{prop:param-specmm}
    For any $\phi\in\calX$, there exist unique real numbers $\theta_1,\theta_2,\theta_3$
    such that
    \[ \phi(\ang{n,m,p}) = n^{\theta_1} m^{\theta_2} p^{\theta_3} \]
    for any $n,m,p$.
\end{proposition}
\begin{proof}
    Let $\theta_1 = \log_2 \phi(\ang{2,1,1})$. We will show that $\phi(\ang{n,1,1})=n^{\theta_1}$ for all $n$.
    For any $k$, let $t = \lfloor \log_2 (n^k) \rfloor$,
    thus $2^t\leq n^k < 2^{t+1}$ and we have the restrictions $\ang{2^t,1,1} \leq \ang{n^k,1,1} \leq \ang{2^{t+1},1,1}$.
    Since $\phi$ is monotone, we have
    \[ \phi(\ang{2^t,1,1}) \leq \phi(\ang{n^k,1,1}) \leq \phi(\ang{2^{t+1},1,1}). \]
    Since $\ang{n^k,1,1} = \ang{n,1,1}^{\otimes k}$ and $\phi$ is multiplicative, we can pull out the exponent:
    \[ \phi(\ang{2, 1,1})^t \leq \phi(\ang{n,1,1})^k \leq \phi(\ang{2,1,1})^{t+1}. \]
    Thus
    \[ 2^{-\theta_1} n^{k \theta_1} \leq 2^{\theta_1 t} \leq \phi(\ang{n,1,1})^k \leq 2^{\theta_1 (t+1)}
    \leq 2^{\theta_1} n^{k\theta_1}. \]
    Taking $k$th roots and letting $k\to \infty$, we conclude $\phi(\ang{n,1,1}) = n^{\theta_1}$.
    Similarly, we can obtain $\theta_2, \theta_3$ such that $\phi(\ang{1,m,1})=m^{\theta_2}$
    and $\phi(\ang{1,1,p})=p^{\theta_3}$.

    Finally, since $\ang{n,m,p}=\ang{n,1,1}\otimes \ang{1,m,1}\otimes\ang{1,1,p}$
    and $\phi$ is multiplicative, we have
    \begin{align*}
        \phi(\ang{n,m,p}) &= \phi(\ang{n,1,1}\otimes \ang{1,m,1}\otimes\ang{1,1,p})\\
        &= \phi(\ang{n,1,1}) \cdot \phi(\ang{1,m,1}) \cdot \phi(\ang{1,1,p}) \\
        &= n^{\theta_1} m^{\theta_2} p^{\theta_3}. \qedhere
    \end{align*}
\end{proof}

Therefore, the behavior of the asymptotic spectrum on matrix multiplication tensors can be
parametrized just by $\theta_1,\theta_2$ and $\theta_3$. This leads to the notion of the 
\emph{spectrum of matrix multiplication}.
\begin{definition}[{\cite[\S 4]{Str88AsymSpec}}]
    The (logarithmic) spectrum $\specMM \subset \bbR^3$ of matrix multiplication
    are the points $(\theta_1,\theta_2,\theta_3)$ that satisfy \Cref{prop:param-specmm}.
\end{definition}

All the basic facts of matrix multiplication tensor transfer to statements of $\specMM$.

\begin{proposition} \label{prop:specmm-sum-theta}
    Every $\theta_i$ lies in $[0, 1]$.
    The value $\theta_1+\theta_2+\theta_3$ over $\specMM$ attains minimum
    $2$ and maximum $\omega$.
\end{proposition}
\begin{proof}
    For any positive integer $n$, since $\AQ(\ang{n,1,1}) = 1$ and $\AR(\ang{n,1,1})=n$, we have
    $\min \theta_1 = 0$ and $\max \theta_1 = 1$, and similarly for $\theta_2,\theta_3$.
    The value $\theta_1 + \theta_2 + \theta_3$, by definition, equals 
    $\log_n \phi(\ang{n,n,n})$. Thus by \Cref{thm:duality}, we have
    \[ \max_{(\theta_1,\theta_2,\theta_3)\in \specMM} (\theta_1 + \theta_2 + \theta_3)
    = \max_{\phi\in\calX} \log_n \phi(\ang{n,n,n}) = \log_n \AR(\ang{n,n,n}), \]
    and it follows from the definition of $\omega$ that $\AR(\ang{n,n,n})=n^\omega$.
    Similarly, we have
    \[ \min_{(\theta_1,\theta_2,\theta_3)\in \specMM} (\theta_1 + \theta_2 + \theta_3) = \log_n \AQ(\ang{n,n,n}), \]
    and matrix multiplication has optimal asymptotic subrank $\AQ(\ang{n,n,n}) = n^2$
    \cite[Theorem 6.6]{Str87RelBilin}.
\end{proof}

\subsection{Toy example} \label{subsec:toy}

\begin{proposition} \label{prop:toy_degen}
    Suppose $\bbF$ is an algebraically closed field.
    For any $3 \times 3 \times 3$ tensor $T$, we have $T \degen \ang{3} \oplus \ang{1,2,1}$.
\end{proposition}

\begin{proof}
    It is known that degeneration can be characterized topologically (cf.~\cite[Theorem 5.8]{Str87RelBilin}, \cite[Theorem 20.24]{BCS13ACT}). Let $C \subseteq \bbF^3 \otimes \bbF^3 \otimes \bbF^3$ be the set of tensors that can be restricted from $\ang{3} \oplus \ang{1,2,1}$; then $T \degen \ang{3} \oplus \ang{1,2,1}$ if and only if $T \in \ov{C}$, where $\ov{C}$ denotes the Zariski closure of $C$. Thus, it suffices to prove $T \leq \ang{3} \oplus \ang{1,2,1}$ for a \emph{generic} $T$, and the statement follows.

    If $T$ is generic, we may first assume that the $z_1$-slice of $T$ is a nonsingular matrix. By a change of basis, we assume that the $z_1$-slice is of the form $x_1 y_1 + x_2 y_2 + x_3 y_3$. Furthermore, we can assume the $z_2$-slice is a diagonalizable matrix; we then apply the corresponding similarity transformation which preserves the form of the first slice and transforms the second slice into $\alpha_1 x_1 y_1 + \alpha_2 x_2 y_2 + \alpha_3 x_3 y_3$. Now we can write the tensor $T$ as
    \[ \sum_{i=1}^3 x_i y_i (z_1 + \alpha_i z_2) + (xMy) z_3 \]
    for some matrix $M$. Letting $\lambda$ be one of the eigenvalues of $M$, then we have
    \[ T = \sum_{i=1}^3 x_i y_i (z_1 + \alpha_i z_2 + \lambda z_3) + x(M - \lambda I)y z_3 \]
    where the first summation can be restricted from $\ang{3}$ and the second term is a matrix of rank at most $2$, which is a restriction of $\ang{1,2,1}$.
\end{proof}

\begin{proposition} \label{prop:toy_bound}
    If tensor $T$ satisfies the degeneration relations on three directions,
    \begin{align}
        T & \degen \ang{r} \oplus \ang{s,1,1}, \label{eq:toy1}\\
        T & \degen \ang{r} \oplus \ang{1,s,1}, \label{eq:toy2}\\
        T & \degen \ang{r} \oplus \ang{1,1,s}, \label{eq:toy3}
    \end{align}
    then its asymptotic rank satisfies $\AR(T) \leq r + s^{\omega/3}$.
\end{proposition}

\begin{proof}
    In this proof, we demonstrate the general principle of \emph{Strassen calculus}.

    \emph{Step 1: Pass to asymptotic spectrum.} Let $\phi \in \calX$ be a point of the asymptotic spectrum. Recall that $\phi$ is \emph{monotone}---it preserves degenerations---so we can apply it to \eqref{eq:toy1}, resulting in
    \[ \phi(T) \leq \phi(\ang{r} \oplus \ang{s,1,1}). \]
    Using the \emph{additivity} of $\phi$, we can pull out the direct sum:
    \[ \phi(T) \leq \phi(\ang{r}) + \phi(\ang{s,1,1}). \]
    Moreover, $\phi$ is \emph{normalized}, so $\phi(\ang{r}) = r$. As shown in \Cref{prop:param-specmm}, the value of a matrix multiplication tensor can be parametrized by $(\theta_1, \theta_2, \theta_3)$ uniquely determined by $\phi$, such that $\phi(\ang{s,1,1}) = s^{\theta_1}$. We can apply the same argument to \eqref{eq:toy2} and \eqref{eq:toy3}, obtaining
    \[ \phi(T) \leq r + s^{\theta_1}, \quad \phi(T) \leq r + s^{\theta_2}, \quad \phi(T) \leq r + s^{\theta_3}. \]

    \emph{Step 2: Use known inequalities to estimate spectral values.} Each of the three degenerations gives an upper bound of $\phi(T)$. To obtain the tightest bound, we take the minimum:
    \[ \phi(T) \leq r + \min \{s^{\theta_1}, s^{\theta_2}, s^{\theta_3}\}. \]
    To translate this into an absolute value independent of the spectral point $\phi$, we utilize \Cref{prop:specmm-sum-theta}: $\theta_1 + \theta_2 + \theta_3 \leq \omega$. We remind the reader that this inequality essentially stems from the fact that $\AR(\ang{n,n,n}) = n^\omega$. Since the $\theta_i$ are all non-negative, at least one must be less than or equal to the average $\omega/3$, thus
    \[ \phi(T) \leq r + s^{(\theta_1 + \theta_2 + \theta_3) / 3} \leq r + s^{\omega / 3}. \]

    \emph{Final step: Return to primal estimate.} Now we have an absolute upper bound of $\phi(T)$ for any $\phi$. By Strassen duality theorem (\Cref{thm:duality}), the asymptotic rank $\AR(T)$ of $T$ is the maximum of $\phi(T)$ over $\calX$, so we get $\AR(T) \leq r + s^{\omega/3}$.
\end{proof}

\begin{corollary} \label{cor:dim_three_ar}
    Any $3 \times 3 \times 3$ tensor $T$ has $\AR(T) \leq 3 + 2^{\omega/3}$.
\end{corollary}

\begin{proof}
    Since asymptotic rank is invariant under field extension \cite[Proposition~15.17]{BCS13ACT}, we may assume $\bbF$ is algebraically closed. Then by \Cref{prop:toy_degen}, we have $T \degen \ang{3} \oplus \ang{1,2,1}$. Applying \Cref{prop:toy_degen} on the tensor with modes relabeled yields the other two degenerations required by \Cref{prop:toy_bound}. Thus we have $\AR(T) \leq 3 + 2^{\omega/3}$.
\end{proof}

Previously, to the best of our knowledge, the upper bound on $\AR(T)$ for a generic $3 \times 3 \times 3$ tensor $T$ was $5$, which is the generic rank. We see that $3 + 2^{\omega/3} < 3 + 2 = 5$ represents an improvement. Numerically, using the current upper bound for $\omega$, we have $3 + 2^{\omega/3} < 4.7297$.

The generic rank (Strassen--Lickteig \cite{Str83GenericR,Lic85GenericR}) establishes that for any $d \times d \times d$ tensor $T$, $T \degen \ang{r_d}$, where
\[ r_d = \begin{cases}
    5, & \text{if } d = 3,\\
    \left\lceil\frac{d^3}{3d-2}\right\rceil, & \text{otherwise}.
\end{cases} \]
This bound works well for small $d$, while for large $d$ the current best asymptotic rank bound comes from $T \leq \ang{d,1,d}$ (though we will refine this in \Cref{subsec:exponent}). This suggests an alternative approach to improving the asymptotic rank of general tensors:

\begin{problem}
    Find an intermediate tensor $S_d$ that admits degenerations for all $d \times d \times d$ tensors (or significant subclasses) to obtain better asymptotic rank estimates.
\end{problem}

We leave this question for future research.

\begin{remark}
    For this example, it is not difficult to derive the same upper bound via \emph{traditional derivations}. We encourage the reader to consider such an alternative proof, which may help in understanding the interplay between classical deductions and the operations in Strassen calculus. However, as the given degenerations become more complicated, the complexity of classical deduction grows rapidly and becomes unintuitive. In contrast, the \emph{Strassen calculus} remains systematic, typically adding only a few lines of derivation. In certain regimes, such as those analyzed in \Cref{sec:omega}, the optimum of the optimization problem may be difficult to describe by an easy criterion, making it reasonable to rely on computer-aided evaluation. In such cases, a traditional derivation would need to incorporate different regimes and sensitively detect shifts in behavior as parameters change, which is a prohibitively complicated task.
\end{remark}

\subsection{Complicated example} \label{subsec:complicated}

In this subsection we present another example of Strassen calculus, based on a slightly more complicated
set of identities that will also be used in later applications.

\begin{proposition} \label{prop:comp_ex}
    If tensor $T$ satisfies the degeneration relations on three directions
    \begin{align}
        T \oplus \ang{t,1,1} & \degen \ang{r} \oplus \ang{s,1,1}, \label{eq:degen1}\\
        T \oplus \ang{1,t,1} & \degen \ang{r} \oplus \ang{1,s,1}, \label{eq:degen2}\\
        T \oplus \ang{1,1,t} & \degen \ang{r} \oplus \ang{1,1,s}, \label{eq:degen3}
    \end{align}
    for some $s,t\ge 1$, then its asymptotic rank satisfies
    \[
        \AR(T) \le
        \begin{cases}
            r + s^{2/3} - t^{2/3}, & \text{if } t>s,\\
            r + s^{\omega/3} - t^{\omega/3}, & \text{if } t\le s.
        \end{cases}
    \]
\end{proposition}

\begin{proof}
Fix an arbitrary spectral point $\phi \in \calX$.
Applying $\phi$ to the three degeneration relations yields:
\[
\renewcommand{\arraystretch}{1.4}
\begin{array}{rcll}
\phi(T)
&\le& r + s^{\theta_1} - t^{\theta_1}
& \text{(apply $\phi$ to \eqref{eq:degen1})} \\[0.3em]
\phi(T)
&\le& r + s^{\theta_2} - t^{\theta_2}
& \text{(apply $\phi$ to \eqref{eq:degen2})} \\[0.3em]
\phi(T)
&\le& r + s^{\theta_3} - t^{\theta_3}
& \text{(apply $\phi$ to \eqref{eq:degen3})}
\end{array}
\]
where $\phi(\ang{r}) = r$ by normalization, and
$\phi(\ang{t,1,1}) = t^{\theta_1}$,
$\phi(\ang{1,t,1}) = t^{\theta_2}$,
$\phi(\ang{1,1,t}) = t^{\theta_3}$
by \Cref{prop:param-specmm}.

Taking the strongest bound among the three directions yields
\[
\renewcommand{\arraystretch}{1.4}
\begin{array}{rcll}
\phi(T)
&\le& r + \min_{i\in\{1,2,3\}} \bigl(s^{\theta_i} - t^{\theta_i}\bigr)
& \text{(take the minimum over directions)} \\[0.4em]
&=& r + \min_{i\in\{1,2,3\}} \tau(\theta_i)
& \text{(set $\tau(x)=s^x-t^x$).}
\end{array}
\]

Now we split into two cases.

\emph{Case 1: $t>s$.} Then $\tau$ is decreasing on $\bbR_{\ge 0}$, so
$\min_i \tau(\theta_i)=\tau(\max_i \theta_i)$.
Since $\theta_1+\theta_2+\theta_3\ge 2$ (\Cref{prop:specmm-sum-theta}), we have $\max_i \theta_i\ge 2/3$, and thus
\[
    \phi(T)\le r+\tau(\max_i \theta_i)\le r+\tau(2/3)=r+s^{2/3}-t^{2/3}.
\]

\emph{Case 2: $t\le s$.} Then $\tau$ is increasing on $\bbR_{\ge 0}$, so
$\min_i \tau(\theta_i)=\tau(\min_i \theta_i)$.
Since $\theta_1+\theta_2+\theta_3\le \omega$ (\Cref{prop:specmm-sum-theta}), we have $\min_i \theta_i\le (\theta_1+\theta_2+\theta_3)/3 \le \omega/3$, and therefore
\[
    \phi(T)\le r+\tau(\min_i \theta_i)\le r+\tau(\omega/3)=r+s^{\omega/3}-t^{\omega/3}.
\]

Since the above inequality holds for every $\phi\in\calX$,
Strassen duality gives
\[
\AR(T) = \max_{\phi\in\calX} \phi(T)
\le
\begin{cases}
    r + s^{2/3} - t^{2/3}, & \text{if } t>s,\\
    r + s^{\omega/3} - t^{\omega/3}, & \text{if } t\le s.
\end{cases}
\qedhere
\]
\end{proof}

\begin{remark}
    The proof illustrates a typical feature of Strassen calculus: it can \emph{turn lower bounds into upper bounds}.
    In the case $t > s$ we used the nontrivial inequality $\theta_1+\theta_2+\theta_3\geq 2$, which ultimately comes from the optimal asymptotic subrank bound $\AQ(\ang{n,n,n})=n^2$.
\end{remark}

\section{Speedup theorems} \label{sec:speedup}

\subsection{Notation and blockwise restrictions}

To clarify our statements, we first introduce notation for handling tensor modes and restrictions. Let $T \in U' \otimes V' \otimes W'$ and $S \in U \otimes V \otimes W$ be tensors. If a tensor mode is identified as a direct sum, such as $U = U_1 \oplus U_2$ and $U' = U'_1 \oplus U'_2$, the linear map $A$ in a restriction $T \leq S$ can be written in blockwise form:
\[ A = \begin{pmatrix} A_{11} & A_{12} \\ A_{21} & A_{22} \end{pmatrix}, \]
where $A_{ij}$ is a linear map from $U_j$ to $U'_i$. Correspondingly, we may decompose the tensor $S$ as $S = S_1 + S_2$ where $S_j \in U_j \otimes V \otimes W$. In this context, the restriction $T = (A \otimes B \otimes C)S$ can be expanded as:
\[ T = \sum_{i,j} (A_{ij} \otimes B \otimes C) (S_j). \]

We also need to clarify the relationship between subspaces of functional spaces and the linear maps used in restrictions. Let $U^\vee$ denote the dual space of $U$, consisting of all linear functions $f \colon U \to \bbF$. Suppose $A \subseteq U^\vee$ is a subspace of dimension $d$. We can identify the entire subspace $A$ with a surjective linear map $U \to \bbF^d$: choosing a basis $\{f_1, \dots, f_d\}$ for $A$ defines the coordinates of the image in $\bbF^d$. While this specific map depends on the choice of basis, any two such choices yield the same resulting tensor up to isomorphism. Thus, we may treat a subspace of $U^\vee$ as an ``aggregate'' linear map in a restriction $(A \otimes B \otimes C) S$ without ambiguity.

\subsection{Free-lunch speedup theorem}

We begin by describing a condition that allows one to obtain a direct sum on the left-hand side \emph{for free}, without the need to add any additional terms to the right-hand side. For readers unfamiliar with these works, it is worth noting that the arguments of
Coppersmith--Winograd \cite[Theorem A]{CW82Improve} and
Strassen \cite[Lemma 3.11]{Str88AsymSpec} likewise proceed in two distinct stages. Our approach follows this established logic, beginning here with the formulation of the first step.

\begin{theorem}[Free-lunch speedup] \label{thm:freelunch_speedup}
    Let $T \in U' \otimes V' \otimes W'$ and $S \in U \otimes V \otimes W$ be two tensors, with a restriction $T \leq S$ given by the maps $(A, B, C)$. Let $C^\perp \subseteq W^\vee$ be the subspace of linear functions $f \colon W \to \bbF$ such that $(A \otimes B \otimes f)S = 0$. For any chosen subspace $C' \subseteq C^\perp$, define the subspaces $A' \subseteq U^\vee$ and $B' \subseteq V^\vee$ as follows:
    \begin{align*}
        A' &= \{ u \in U^\vee : (u \otimes B \otimes C')S = 0\}, \\
        B' &= \{ v \in V^\vee : (A \otimes v \otimes C')S = 0\}.
    \end{align*}
    Then we have the direct sum degeneration:
    \[ T \oplus T' \degen S, \]
    where $T' = (A' \otimes B' \otimes C') S$.
\end{theorem}

\begin{proof}
    The construction of the degeneration proceeds in two steps.

    \emph{Step 1.} We first construct a restriction by augmenting the linear maps. We define the new maps as $(A,A')$, $(B,B')$, and $(C,C')$, where the block structure corresponds to the direct sum of the target spaces. Consider the resulting tensor:
    \[ R = \left[ \begin{pmatrix} A \\ A' \end{pmatrix} \otimes \begin{pmatrix} B \\ B' \end{pmatrix} \otimes \begin{pmatrix} C \\ C' \end{pmatrix} \right] S. \]
    We examine the blocks of $R$ arising from the terms involving $C'$. By the construction of our subspaces:
    \begin{itemize}
        \item $(A \otimes B \otimes C') S = 0$ by the definition of $C'$.
        \item $(A' \otimes B \otimes C') S = 0$ by the definition of $A'$.
        \item $(A \otimes B' \otimes C') S = 0$ by the definition of $B'$.
    \end{itemize}
    Consequently, the only potentially non-zero block involving $C'$ is $(A' \otimes B' \otimes C') S = T'$. In the overall block structure of $R$, the diagonal blocks are $R_{111} = T$ and $R_{222} = T'$, while the blocks $R_{112}, R_{122},$ and $R_{212}$ are guaranteed to be zero.

    \emph{Step 2.} We now eliminate the remaining junk terms (namely $R_{121}, R_{211},$ and $R_{221}$) via a \emph{monomial degeneration}. We apply the following diagonal transformation to the tensor $R$:
    \[ \begin{pmatrix} 1 & \\ & \lambda \end{pmatrix} \otimes \begin{pmatrix} 1 & \\ & \lambda \end{pmatrix} \otimes \begin{pmatrix} \lambda^2 & \\ & 1 \end{pmatrix} R = \lambda^2 (T + T') + O(\lambda^3). \]
    By scaling by $\lambda^{-2}$ and taking the limit as $\lambda \to 0$, the junk terms are eliminated, yielding the direct sum $T \oplus T'$. In conclusion, we have obtained $T \oplus T' \degen R \leq S$, which implies $T \oplus T' \degen S$.
\end{proof}

One might be concerned about why \Cref{thm:freelunch_speedup} only speaks about starting with a \emph{restriction} rather
than \emph{degeneration}. This is because the degeneration version is obtained by bootstrapping
the result: It is not using any specific property of degeneration, but just regarding it
as a restriction over the field extension $\bbF(\lambda)$.
\begin{corollary}[Free-lunch speedup, degeneration version] \label{cor:freelunch_speedup_degen}
    Let $T \degen S$ be a degeneration of $\bbF$-tensors. More specifically,
    let $T_\lambda$ be a $\bbF(\lambda)$-tensor before taking $\lambda\to 0$, 
    i.e.,
    \[  T + O(\lambda) = T_\lambda = (A_\lambda \otimes B_\lambda \otimes C_\lambda)S \]
    is the restriction of $\bbF(\lambda)$-tensors.
    Plug in $T \leq S$ with data $A_\lambda, B_\lambda$ and $C_\lambda$ into \Cref{thm:freelunch_speedup},
    for any choice of $C'$ thereby,
    let the resulting $\bbF(\lambda)$ tensor be $T'_\lambda$, $T'$ be the lowest order term of $T'_\lambda$ ($0$ if $T'_\lambda$ is
    completely zero). Then
    \[ T \oplus T' \degen S. \]
\end{corollary}
\begin{proof}
    Firstly, by \Cref{thm:freelunch_speedup}, we obtain a degeneration of $\bbF(\lambda)$-tensors
    \[ T_\lambda \oplus T'_\lambda \degen S, \]
    say this degeneration is parametrized by another symbol $\varepsilon$. We may scale by $\lambda$ and $\varepsilon$,
    to assume that the degeneration is given by $A_{\lambda, \varepsilon},
    B_{\lambda, \varepsilon}$ and $C_{\lambda, \varepsilon}$, and $T'_\lambda = T' + O(\lambda)$,
    such that
    \[ (T \oplus T') + O(\lambda) + O_\lambda(\varepsilon) = (A_{\lambda, \varepsilon}
    \otimes B_{\lambda, \varepsilon}
    \otimes C_{\lambda, \varepsilon})S. \]
    The notation means that $O(\lambda)$ term is purely a $\lambda$-high order term tensor over $\bbF$,
    but $O_\lambda(\varepsilon)$ term is a $\varepsilon$-high order term tensor that can also contain $\lambda$.
    One can again reduce to the regular definition of degeneration by substituting $\varepsilon$
    by a large enough power $\varepsilon = \lambda^k$, then the whole restriction again becomes
    \[ (T \oplus T') + O(\lambda) = (A_{\lambda, \lambda^k}
    \otimes B_{\lambda, \lambda^k}
    \otimes C_{\lambda, \lambda^k})S. \]
    This concludes the degeneration.
\end{proof}

\subsubsection{Isolated summand}

We now introduce the notion of an \emph{isolated} summand, a property defined by Coppersmith and Winograd \cite[Definition 3.1]{CW82Improve} that is crucial for iterating speedup theorems.

Let $T \oplus T' \degen S$ be a degeneration defined by the polynomial maps $A_\lambda, B_\lambda$, and $C_\lambda$. Let $R_\lambda = (A_\lambda \otimes B_\lambda \otimes C_\lambda)S$ denote the $\bbF(\lambda)$-tensor prior to taking the limit $\lambda \to 0$. By definition, we must have $(R_\lambda)_{111} = T + O(\lambda)$ and $(R_\lambda)_{222} = T' + O(\lambda)$, while all other blocks $(R_\lambda)_{ijk}$ (where indices are not all $1$ or all $2$) must satisfy $(R_\lambda)_{ijk} = O(\lambda)$. We say that the direct summand $T'$ is an \emph{isolated summand} in the third mode with respect to the degeneration $(A_\lambda, B_\lambda, C_\lambda)$ if $(R_\lambda)_{ij2}$ is identically zero for all $(i, j) \neq (2, 2)$.

One can verify that this definition essentially coincides with Coppersmith and Winograd's original formulation when $T' = \ang{1,t,1}$. Moreover, the following facts are readily established.

\begin{proposition} \label{prop:speedup_isolated}
    The direct summand $T'$ of $T \oplus T'\degen S$ produced by \Cref{thm:freelunch_speedup} or \Cref{cor:freelunch_speedup_degen} is isolated in the third mode.
\end{proposition}

The following proposition generalizes \cite[Proposition 4.1]{CW82Improve}:
\begin{proposition} \label{prop:prod_isolated}
    Let $T_1', T_2'$ be isolated direct summands of degenerations
    \[ T_1 \oplus T_1' \degen S_1, \quad T_2 \oplus T_2' \degen S_2. \]
    Then, in the tensor product of these degenerations
    \[ (T_1 \otimes T_2) \oplus (T_1\otimes T_2') \oplus (T_1'\otimes T_2) \oplus (T_1' \otimes T_2') \degen S_1 \otimes S_2, \]
    the term $T_1'\otimes T_2'$ remains an isolated summand in the third mode.
\end{proposition}

\subsection{Speedup by one slice}

In this subsection, we recover the classical speedup theorem of Coppersmith and Winograd, as well as Strassen's generalization, in the case where the tensors appended to both sides are one-slice matrix multiplication tensors. We also refine the notion of \emph{fullness} originally introduced by Coppersmith and Winograd \cite[Definition 3.1]{CW82Improve}.

Usually, one cannot hope to directly use the free-lunch speedup theorem, because the first step requires choosing a subspace $C' \subseteq C^\perp$. However, for a reasonable degeneration $T \degen S$, the corresponding space $C^\perp$ is typically zero. Indeed, if $C^\perp$ were non-zero, then viewing $S$ as an ``upper bound'' for $T$, the tensor $S$ would contain redundant $W$-slices; removing these slices would essentially yield a more effective upper bound for $T$.

Therefore, the standard speedup strategy proceeds as follows: First, artificially add a tensor $S'$ to the right-hand side. This new summand $S'$ contributes nothing to the degeneration $T \degen S \oplus S'$ because its component in the $C$-map is set to zero. However, its components in the $A$ and $B$ maps can be carefully chosen to satisfy the conditions of \Cref{thm:freelunch_speedup} (or \Cref{cor:freelunch_speedup_degen}), enabling the extraction of a non-trivial summand $T'$.

We first examine the implications of the free-lunch speedup theorem when $C'$ is a one-dimensional subspace.

\begin{proposition} \label{prop:freelunch_oneslice}
    Assume the setting of \Cref{thm:freelunch_speedup}. Let $C' \subseteq C^\perp$ be a one-dimensional subspace spanned by a linear function $f \colon W \to \bbF$. Let $r = \Rk((\id_U \otimes \id_V \otimes f) S)$ be the rank of the resulting matrix in $U \otimes V$. Then we have $T' \cong \ang{1, t, 1}$, where
    \[ t \geq r - n - m, \]
    with $n = \dim U'$ and $m = \dim V'$.
\end{proposition}

\begin{proof}
    The restriction $T' = (A' \otimes B' \otimes f)S$ can be viewed as first contracting the $W$-mode: $M = (\id_U \otimes \id_V \otimes f)S$ yields a matrix of rank $r$. The subspace $A'$ is defined as the kernel of the map $u \mapsto (u \otimes B \otimes f)S$. Since this map maps to $V'$, its kernel $A'$ has codimension at most $\dim V' = m$. Thus, restricting $M$ to $A'$ reduces the rank by at most $m$. Similarly, restricting to $B'$ reduces the rank by at most $n$. Consequently, $T' = (A' \otimes B')M$ has rank at least $t \geq r - n - m$. A rank-$t$ matrix (viewed as a tensor) is isomorphic to $\ang{1,t,1}$.
\end{proof}

We now explain how to append a one-slice tensor to both sides of the degeneration.

\begin{proposition}[One-slice speedup {\cite[Lemma 3.12]{Str88AsymSpec}}] \label{prop:add_slice}
    Let $T \leq S$ be a restriction defined by maps $A, B, C$. For a linear function $f \colon W \to \bbF$, let $q = \Rk((\id_U \otimes \id_V \otimes f)S)$ be the rank of the corresponding matrix in $U \otimes V$, and let $s = \Rk((A \otimes B \otimes f)S)$ be the rank of the matrix in $U' \otimes V'$. Then there exists a restriction $T \leq S \oplus \ang{1,s,1}$ extending $A, B, C$, such that there exists a linear function $f' \colon (W \oplus \bbF) \to \bbF$ (where $W \oplus \bbF$ corresponds to the third mode of $S \oplus \ang{1,s,1}$) satisfying the conditions of \Cref{prop:freelunch_oneslice} with $r = q+s$.
\end{proposition}

\begin{proof}
    Since $M = (A \otimes B \otimes f)S$ has rank $s$, there exist linear maps $A' \colon \bbF^s \to U'$ and $B' \colon \bbF^s \to V'$ such that $M = (A' \otimes B') \ang{1,s,1}$. It is then straightforward to verify that we can define the new restriction as
    \[ T = \left[\begin{pmatrix}
        A & A'
    \end{pmatrix} \otimes \begin{pmatrix}
        B & B'
    \end{pmatrix} \otimes \begin{pmatrix}
        C & 0
    \end{pmatrix}\right] (S \oplus \ang{1,s,1}), \]
    with the new linear function $f' = \begin{pmatrix} f & -1 \end{pmatrix}$. This function lies in the appropriate kernel, since
    \[ \left[ \begin{pmatrix}
        A & A'
    \end{pmatrix} \otimes \begin{pmatrix}
        B & B'
    \end{pmatrix} \otimes \begin{pmatrix}
        f & -1
    \end{pmatrix} \right] (S \oplus \ang{1,s,1})
    = M - (A'\otimes B')\ang{1,s,1} = 0, \]
    and $({\id} \otimes {\id} \otimes f')(S \oplus \ang{1,s,1})$ is the direct sum of the matrices $({\id_U} \otimes {\id_V} \otimes f)S$ and $\ang{1,s,1}$, which has rank $r = q + s$.
\end{proof}

Finally, we generalize Coppersmith and Winograd's notion of \emph{fullness}. Given a degeneration $T \oplus \ang{1,t,1} \degen S$, let $f \colon W \to \bbF(\lambda)$ be the component of the $C$-map projecting to the slice $\ang{1,t,1}$. We define the \emph{fullness index} of $\ang{1,t,1}$ (or $f$) as $q = \Rk ((\id_U \otimes \id_V \otimes f)S)$. If $q = \dim U = \dim V$, we say the slice is \emph{full}.

Like the isolated property, the fullness index behaves well under speedup theorems and tensor products.

\begin{proposition} \label{prop:speedup_full}
    The slice $\ang{1,t,1}$ constructed in \Cref{prop:freelunch_oneslice} has fullness index $r$.
\end{proposition}

\begin{proposition} \label{prop:prod_full}
    Consider two degenerations
    \[ T_1 \oplus \ang{1,t_1, 1} \degen S_1, \quad T_2 \oplus \ang{1, t_2, 1} \degen S_2, \]
    wherein the summands $\ang{1,t_1,1}$ and $\ang{1,t_2,1}$ have fullness indices $q_1$ and $q_2$, respectively. Then the direct summand $\ang{1, t_1 t_2, 1}$ in the tensor product of these two degenerations has fullness index $q_1 q_2$. In particular, if $f_1$ and $f_2$ are full, then $f_1 \otimes f_2$ is full.
\end{proposition}

\subsection{Direct sum and compression of matrix multiplication}

In the proof of \Cref{prop:freelunch_oneslice}, we used the following basic fact:
for a one-slice matrix multiplication tensor $\ang{1,r,1}$, restricting its first mode to a
codimension-$p$ subspace yields a tensor isomorphic to $\ang{1,r-p,1}$.
Although this observation looks elementary, it can be viewed as the simplest instance of a much more
general phenomenon.

Strassen proved the following result (see \cite[Proposition 6.4]{Str88AsymSpec})\footnote{Strassen's original formulation and subsequent expositions (e.g., \cite[Proposition~6.4]{Str88AsymSpec} and \cite[Part~III]{WZ25Spectra}) often present a closely related variant in terms of an \emph{additive} perturbation in the same tensor space: if $S=T+U$ where $U$ has small flattening rank in the $U$-mode, then one can find a ``valuable'' restriction $T'\le S$ that remains a direct sum of matrix multiplication tensors. The formulation we use here (starting from a direct sum and restricting one mode by a small-codimension subspace) is equivalent to that variant via the substitution method; we include this footnote to reconcile the two viewpoints.}:
if $T$ is a direct sum of matrix multiplication tensors and $A\subseteq U^\vee$ is a subspace of small
codimension, then the partially restricted tensor
$(A\otimes \id_V\otimes \id_W)T$ still contains, as a restriction, a ``valuable'' tensor $T'$ that is
again a direct sum of matrix multiplication tensors.
Informally, matrix multiplication is robust under losing a small amount of information in one mode.
This is often referred to as the \emph{compression} property of matrix multiplication tensors; see
\cite[Part~III]{WZ25Spectra} for an exposition and generalizations.
Algorithmic versions have also been developed recently \cite{HS25CorrectionA,HS25CorrectionB}.

For our purposes, we only need a compression statement in the special case where all summands are
one-slice tensors.
This case is simple enough that we include a short proof for the reader's convenience.
We also note that the concrete applications in this paper do not require the full generality of this theorem:
the tensors and linear maps we use enjoy additional structure that simplifies the analysis. Nonetheless, we expect
that the compression statement in this general form may be useful for more complicated degenerations in the future,
for instance when the relevant linear maps are found via computer search.

\begin{proposition} \label{prop:compression}
    Let $T = \bigoplus_{i=1}^k \ang{1,n_i,1}$ be a direct sum of one-slice matrix multiplication tensors, and
    let $A\subseteq U^\vee$ be a subspace of codimension $p$.
    Then there exists a tensor
    \[ T' \le (A\otimes \id_V\otimes \id_W)T \]
    that is still a direct sum of one-slice matrix multiplication tensors:
    \[ T' \cong \bigoplus_{i=1}^k \ang{1,n_i-p_i,1}, \]
    where the integers $p_i$ satisfy $0\le p_i\le n_i$ and $\sum_{i=1}^k p_i = p$.
\end{proposition}

\begin{proof}
    Let $(x^{(i)}_j)_{1\le j\le n_i}$ and $(y^{(i)}_j)_{1\le j\le n_i}$ denote the bases of the
    $U$- and $V$-modes of $T$ corresponding to the decomposition
    $T=\bigoplus_{i=1}^k \ang{1,n_i,1}$.
    Identify $A$ with a surjective linear map $A:U\to U'$ where $\dim U' = \dim U - p$.

    Consider the set of vectors $\{A(x^{(i)}_j)\} \subseteq U'$.
    Choose a subset $D\subseteq \{x^{(i)}_j\}$ of size $\dim U'$ such that
    $\{A(x): x\in D\}$ is linearly independent (hence a basis of $U'$).
    Now restrict $T$ by zeroing out every basis element $x^{(i)}_j$ (and the corresponding
    $y^{(i)}_j$) with $x^{(i)}_j\notin D$.
    This reduces the $i$th summand $\ang{1,n_i,1}$ to $\ang{1,n_i-p_i,1}$, where $p_i$ is the number of
    removed basis elements from the $i$th block.
    By construction, $\sum_i p_i = p$.

    Finally, the restriction of $A$ to $\Span D$ is an isomorphism onto $U'$, so after the above
    zeroing-out operation the resulting tensor is indeed a restriction of
    $(A\otimes \id_V\otimes \id_W)T$ and has the claimed direct-sum form.
\end{proof}

\section{Instantiations of speedup theorems} \label{sec:instantiation}

The first method is the most naive version of the speedup theorem, which generalizes \cite[Corollary~3.1]{CW82Improve}.

\begin{theorem}[One-slice speedup for nonminimal border rank] \label{thm:app-oneslice}
    Let $T$ be an $n\times n\times n$ tensor over $\bbF$, and suppose that $T\degen \ang{r}$.
    Let $\{a_i\}, \{b_i\}, \{c_i\}$ be the $\bbF(\lambda)$-vectors representing the border rank decomposition, i.e.,
    \[ \sum_{i=1}^r a_i b_i c_i = T + O(\lambda). \]
    If there exist \emph{nonzero scalars} $c_i' \in \bbF(\lambda) \setminus \{0\}$ such that the matrix
    \[ M = \sum_{i=1}^r a_i b_i c_i' \in U' \otimes V' \]
    has rank at most $s$, then there exists a degeneration
    \[ T \oplus \ang{1, r + s - 2n, 1} \degen \ang{r} \oplus \ang{1, s, 1}. \]
    When $r > 2n$, this yields the improved asymptotic rank bound
    \[ \AR(T) \le r + s^{2/3} - (r+s - 2n)^{2/3} < r. \]

    In particular, if $r\geq n$, we have
    \[ T \oplus \ang{1, r-n, 1} \degen \ang{r} \oplus \ang{1, n, 1}, \]
    and for $r > 2n$, this provides the strict improvement
    \[ \AR(T) \le r + n^{2/3} - (r - n)^{2/3} < r. \]
\end{theorem}

\begin{proof}
    The hypothesis exactly matches the condition of \Cref{prop:add_slice}. Applying that result, all the degeneration parameters align:
    \[ T \oplus \ang{1, r + s - 2n, 1} \degen \ang{r} \oplus \ang{1, s, 1}. \]

    We refer to the setting of the three-direction speedup inequalities \Cref{eq:degen1}--\Cref{eq:degen3} with parameters $t = r+s-2n$ and the same $s$. Since $r > 2n$, we have $t > s$, and thus \Cref{eq:degen1}--\Cref{eq:degen3} imply the asymptotic rank bound
    \[ \AR(T) \le r + s^{2/3} - t^{2/3} = r + s^{2/3} - (r+s-2n)^{2/3}. \qedhere \]
\end{proof}

We can assess the efficacy of this method by the order of magnitude of the improvement. For $r \geq 3n$, the bound simplifies to
\[ \AR(T) \le r - \Omega(r^{2/3}). \]

Although this method appears to yield an improvement only when $r > 2n$, we can bootstrap this result to show that any non-minimal border rank upper bound for $T$ is not tight for its asymptotic rank.

\begin{corollary}
    Let $T$ be an $n\times n\times n$ tensor with $\BR(T) \leq r$. If $r > n$, then $\AR(T) < r$.
\end{corollary}

\begin{proof}
    Since $T \degen \ang{r}$, taking tensor powers yields $T^{\otimes k} \degen \ang{r^k}$. By choosing $k$ sufficiently large such that $r^{k} > 2n^k$, \Cref{thm:app-oneslice} implies $\AR(T)^k = \AR(T^{\otimes k}) < r^k$.
\end{proof}

The second method demonstrates that the identity obtained via the first method can essentially always be \emph{iterated} to yield a superior identity. This generalizes \cite[Theorem B]{CW82Improve}.

\begin{theorem} \label{thm:iterate}
    Under the setting of \Cref{thm:app-oneslice}, let $t = r+s-2n$. There exists a degeneration
    \[ T^{\otimes 2} \; \oplus \; 2 \odot T\otimes \ang{1,t,1} \; \oplus \;
    \ang{1, t^2 + 2n^2, 1}
    \; \degen \; \ang{r^2} \; \oplus \; 2r\odot \ang{1,s,1} \; \oplus \; \ang{1,s^2,1}. \]
    This effectively replaces the summand $\ang{1,t^2,1}$ in the tensor square of the result from \Cref{thm:app-oneslice} with $\ang{1,t^2+2n^2,1}$.
\end{theorem}

\begin{proof}
    Since the summand $\ang{1,t,1}$ appearing in the degeneration
    \[ T \oplus \ang{1,t,1} \degen \ang{r} \oplus \ang{1,s,1} \]
    is produced via the one-slice speedup theorem (\Cref{thm:app-oneslice}), \Cref{prop:speedup_isolated} and \Cref{prop:speedup_full} guarantee that this summand is full and isolated w.r.t.~the degeneration.
    Upon taking the tensor square, \Cref{prop:prod_isolated} and \Cref{prop:prod_full} guarantee that the summand $\ang{1,t^2,1}$ in the degeneration
    \[ T^{\otimes 2} \; \oplus \; 2 \odot T\otimes \ang{1,t,1} \; \oplus \;
    \ang{1, t^2, 1}
    \; \degen \; \ang{r^2} \; \oplus \; 2r\odot \ang{1,s,1} \; \oplus \; \ang{1,s^2,1} \]
    remains full and isolated. Let this summand be projected by a linear functional $f$ in the third mode.
    First, we zero out the variables of this summand in all three directions, so that for the remaining degeneration
    \[ T^{\otimes 2} \; \oplus \; 2 \odot T\otimes \ang{1,t,1}
    \; \degen \; \ang{r^2} \; \oplus \; 2r\odot \ang{1,s,1} \; \oplus \; \ang{1,s^2,1}, \]
    $f$ satisfies the condition of \Cref{prop:freelunch_oneslice}, as $f$ is isolated.
    The parameter $r$ corresponds to the fullness of $f$, which in our case (since $f$ is full) is the dimension of the right-hand side, i.e., $(r+s)^2$.
    Thus the replaced direct summand
    \[ T^{\otimes 2} \; \oplus \; 2 \odot T\otimes \ang{1,t,1} \; \oplus \;
    \ang{1, t', 1}
    \; \degen \; \ang{r^2} \; \oplus \; 2r\odot \ang{1,s,1} \; \oplus \; \ang{1,s^2,1} \]
    has
    \[ t' = (r+s)^2 - 2(n^2 + 2nt) = t^2 + 2n^2. \qedhere \]
\end{proof}

Naturally, this iterative process can be continued indefinitely, but its improvement on the asymptotic rank would be marginal.
Also, after some analysis, it turns out that \Cref{thm:iterate} still gives an asymptotic rank bound
of type $r - \Omega(r^{2/3})$ when $r\geq 3n$.

The third method involves appending not merely a single one-slice tensor to both sides, but a direct sum of such tensors.

\begin{theorem} \label{thm:speedup_via_compression}
    Let $T\degen \ang{r}$, and let $\{a_i\}, \{b_i\}, \{c_i\}$ be $\bbF(\lambda)$-vectors representing the degeneration. Suppose there exist nonzero scalars $c_i'$ and a partition of $[r]$ into $p$ groups $[r] = I_1 \sqcup \cdots \sqcup I_p$. For each $1\leq \alpha\leq p$, let $r_\alpha = |I_\alpha|$, let $M_\alpha = \sum_{i\in I_\alpha} a_i b_i c'_i$, and let $s_\alpha \geq \Rk(M_\alpha)$.
    If $r_\alpha + s_\alpha \geq 2n$ for all $\alpha$, then
    \[ T \oplus \bigoplus_{\alpha=1}^p \ang{1,r_\alpha + s_\alpha - 2n,1}
    \degen \ang{r} \oplus \bigoplus_{\alpha=1}^p \ang{1, s_\alpha,1}. \]
\end{theorem}

We provide two proofs of this result. They are essentially equivalent; the first is simpler, while the second employs the concept of the compression theorem and generalizes to cases where the added terms are not direct sums of one-slice tensors (\Cref{thm:direct_sum_identity}) or where the starting upper bound $S$ is not a diagonal tensor (\Cref{subsec:exponent}).

\begin{proof}[First proof]
    In the spirit of \Cref{cor:freelunch_speedup_degen}, it suffices to prove the restriction case, i.e., working over $\bbF$ instead of $\bbF(\lambda)$.
    In that setting, we can write $T$ as a sum of tensors $T = \sum_{\alpha=1}^p T_\alpha$, where $T_\alpha = \sum_{i\in I_\alpha} a_i b_i c_i$. Then $\Rk(T_\alpha) \leq r_\alpha$. We can apply \Cref{thm:app-oneslice} to each $T_\alpha$ individually and sum the results:
    \[ \bigoplus_{\alpha=1}^p (T_\alpha \oplus \ang{1,r_\alpha + s_\alpha - 2n,1})
    \degen \ang{r} \oplus \bigoplus_{\alpha=1}^p \ang{1, s_\alpha,1}. \]
    The conclusion follows from $T \leq \bigoplus_\alpha T_\alpha$.
\end{proof}

\begin{proof}[Second proof]
    As in the first proof, we focus on the restriction case. We begin with the redundant restriction
    \[ T \leq \ang{r} \oplus \bigoplus_{\alpha=1}^p \ang{1,s_\alpha,1}. \]
    Write
    \[ S \coloneq \ang{r} \oplus \bigoplus_{\alpha=1}^p \ang{1,s_\alpha,1} = \sum_{i=1}^r u_iv_iw_i
     - \sum_{\alpha=1}^p \sum_{j=1}^{s_\alpha} u'_{\alpha j} v'_{\alpha j} w'_\alpha, \]
    where the minus sign is introduced to simplify subsequent cancellations.
    Since $M_\alpha$ has rank at most $s_\alpha$, there exist linear maps $A_\alpha, B_\alpha$ such that $(A_\alpha \otimes B_\alpha) \ang{1,s_\alpha,1} = M_\alpha$.
    The restriction is defined by the linear maps
    \begin{align*}
        A: & \qquad  a_i \mapsfrom u_i, \quad A_\alpha(u'_{\alpha j}) \mapsfrom u'_{\alpha j},\\
        B: & \qquad  b_i \mapsfrom v_i, \quad B_\alpha(v'_{\alpha j}) \mapsfrom v'_{\alpha j},\\
        C: & \qquad  c_i \mapsfrom w_i, \qquad \qquad 0 \mapsfrom w'_\alpha.
    \end{align*}
    To apply the free-lunch speedup theorem (\Cref{thm:freelunch_speedup}), define a subspace $C'\subseteq W^\vee$ via the linear map defined as follows, where $\alpha(i)$ is the index such that $i\in I_\alpha$:
    \[ C': \quad c'_i w'_{\alpha(i)} \mapsfrom w_i, \quad w'_{\alpha} \mapsfrom w'_\alpha. \]
    A direct computation shows that $(A\otimes B\otimes C') S = 0$.

    Applying \Cref{thm:freelunch_speedup} yields a direct summand $T' = (A'\otimes B'\otimes C')S$ on the left-hand side. Let us first unwind the form of $({\id} \otimes {\id} \otimes C')S$:
    \begin{align*}
        ({\id} \otimes {\id} \otimes C')S &= \sum_{i=1}^r u_iv_i \cdot c'_i w'_{\alpha(i)}
        - \sum_{\alpha=1}^p \sum_{j=1}^{s_\alpha} u'_{\alpha j} v'_{\alpha j} w'_\alpha\\
        &= \sum_{\alpha=1}^p \left( \sum_{i\in I_\alpha} c_i' u_i v_i - \sum_{j=1}^{s_\alpha}
        u'_{\alpha j} v'_{\alpha j} \right)w'_\alpha,
    \end{align*}
    which is evidently $\bigoplus_{\alpha=1}^p \ang{1, r_\alpha + s_\alpha, 1}$.
    Recall that $A' \subseteq U^\vee$, the space of linear functionals, is given by those $f\in U^\vee$ such that the following expression vanishes:
    \[ (f\otimes B\otimes C') S = \sum_{\alpha=1}^p \left( \sum_{i\in I_\alpha} c_i' f(u_i) b_i - \sum_{j=1}^{s_\alpha}
        f(u'_{\alpha j}) B_\alpha(v'_{\alpha j}) \right)w'_\alpha \]
    The image, viewed as a vector in an $np$-dimensional space, defines $np$ linear equations on $f$, meaning $A'$ is a subspace of codimension $\leq np$.
    Similarly, $B'$ restricts to a subspace of codimension $\leq np$. \Cref{prop:compression} applies, which allows us to reduce the direct sum of matrix multiplication tensors by $2np$ terms.
    Notice that each equation associated with a component in $w'_\alpha$ only concerns the values of $f$ on $\{u_i\}_{i\in I_\alpha}$ and $u'_{\alpha j}$. Thus, we do not need the full generality of \Cref{prop:compression}, but can instead perform compression on each direct summand separately. Consequently, $T'$ can be compressed to $T' \geq \bigoplus_{\alpha} \ang{1, r_\alpha + s_\alpha - 2n, 1}$, proving the claim.
\end{proof}

We now discuss the quantitative advantage of this method compared with previous ones. Using the naive asymptotic assumption $s\leq n$, one can take $p = \Omega(r / n)$ and let every $r_\alpha \geq 3n$. This gives the degeneration
\[ T \oplus p \odot\ang{1,2n,1} \degen \ang{r} \oplus p\odot\ang{1,n,1}, \]
and similar statements hold for the other two directions.
Strassen calculus then yields an upper bound $\AR(T) \leq r - (2^{2/3} - 1) p n^{2/3} = r - \Omega(r / n^{1/3})$, which is a superior bound when $r = \omega(n)$.

\begin{remark} \label{rmk:generic_ar}
    An intriguing question is how far we can improve the order of magnitude in the speedup theorem. For example, can we prove that any border rank upper bound $\BR(T) \leq r$ with $r > n^{1.0001}$ implies $\AR(T) \leq 0.9999r$? Any mild strengthening of this question would have surprising consequences in fine-grained complexity theory: the \emph{balanced partition tensor} $T_n$ proposed by Pratt \cite{Pratt24SCC} is a family of tensors for which the currently known border rank upper bound is $r_n\approx N^{1.08}$ \cite[Proposition~6.18]{FJM25SymPow}, where $N$ is the dimension of $T_n$. If one could improve the asymptotic rank of $T_n$ to $o(r_n / \log r_n)$, then Pratt showed that this would refute the set cover conjecture.
\end{remark}

\section{Applications} \label{sec:appl}

\subsection{Small CW tensors} \label{subsec:smallcw}

Recall that the small Coppersmith--Winograd (CW) tensor $\cw_q$ has $q+1$ variables in each mode.
It is commonly written as
\[ \cw_q = \sum_{i=1}^q x_0 y_i z_i + x_i y_0 z_i + x_i y_i z_0. \] 
The following identity, originally due to Coppersmith--Winograd, gives a border-rank upper bound of $q+2$.
\begin{align*}
    &\phantom{=}\sum_{i=1}^q \lambda^{-2} (x_0 + \lambda x_i)(y_0 + \lambda y_i)(z_0 + \lambda z_i)\\
    &- \lambda^{-3} \left(x_0 + \lambda^2\sum_{i=1}^q x_i\right)
    \left(y_0 + \lambda^2\sum_{i=1}^q y_i\right)
    \left(z_0 + \lambda^2\sum_{i=1}^q z_i\right)\\
    &+ (\lambda^{-3} - q \lambda^{-2}) x_0 y_0 z_0 \\
    &= \cw_q + O(\lambda).
\end{align*}

It is known that this border-rank bound is tight, i.e.\ $\BR(\cw_q) = q+2$. Moreover, Conner, Gesmundo,
Landsberg, and Ventura \cite[Theorem 1.4]{CGLV22CW} proved that the border rank does not drop for small
powers: when $q>2$ and $n=2$, one has $\BR(\cw_q^{\otimes n}) = (q+2)^n$, and likewise when $q>4$ and $n=3$.
We show that the asymptotic rank of $\cw_q$ is strictly smaller than $q+2$, and hence this phenomenon must
eventually fail for larger $n$.

We begin by finding low-rank matrices in order to apply \Cref{thm:app-oneslice}.
\begin{lemma} \label{lem:minrk}
    Let $a_i,b_i,c_i$ be the data of the border-rank upper bound $\cw_q\degen \ang{q+2}$ given above.
    Then there exist nonzero scalars $c_i'$ such that the matrix
    $M = \sum_{i=1}^{q+2} c_i' \cdot a_i b_i \in U'\otimes V'$ has rank $q$.
\end{lemma}
\begin{proof}
    We first write the linear combination as
    \[ M=\sum_{i=1}^q c_i' (x_0 + \lambda x_i) (y_0 + \lambda y_i)
    + c_{q+1}' \left(x_0 + \lambda^2 \sum_{i=1}^q x_i\right)
    \left(y_0 + \lambda^2 \sum_{i=1}^q y_i\right) + c_0' x_0 y_0. \]
    For notational convenience, we apply the change of basis $x_i' \coloneq x_0 + \lambda x_i$
    for $i > 0$ (and similarly for $y$). Then
    \[ M = c'_0 x_0 y_0 + \sum_{i=1}^q c_i' x_i' y_i' + c'_{q+1}
    \left((1-q \lambda) x_0 + \lambda \sum_{i=1}^q x_i'\right)
    \left((1-q \lambda) y_0 + \lambda \sum_{i=1}^q y_i'\right). \]
    Assume $c_i' = 1$ for all $i \leq q$. We choose $c'_{q+1}$ so that $M$ becomes singular.
    Writing $M$ in the form $M = I + c'_{q+1} ab^\sfT$, we have
    \[ \det (M) = \det(I + c'_{q+1} ab^\sfT) = 1 + c'_{q+1} b^\sfT a, \]
    and a direct calculation gives $b^\sfT a = (1-q\lambda)^2 + q\lambda^2 \neq 0$.
    Taking $c'_{q+1} = -1/(b^\sfT a)$ yields $\Rk(M) < q+1$. On the other hand,
    \[
        \Rk(M) = \Rk(I + c'_{q+1} ab^\sfT) \geq \Rk(I) - \Rk(c'_{q+1} ab^\sfT) = q.
    \]
    Therefore, this choice produces a rank-$q$ matrix.
\end{proof}

\begin{proposition}
    For any $q\geq 2$ and $n\geq 1$, we have the tensor degeneration
    \[ \cw_q^{\otimes n} \oplus \ang{1, (q+2)^n - 2(q+1)^n + q^n, 1}
        \degen \ang{(q+2)^n} \oplus \ang{1,q^n,1}.  \]
\end{proposition}
\begin{proof}
    Let $f$ be the linear functional on the $W$-mode given by the previous lemma. Then
    $M = (A\otimes B\otimes f)\ang{q+2}$ satisfies $\Rk(M) = q$.
    Tensoring these maps gives a degeneration $\cw_q^{\otimes n} \degen \ang{(q+2)^n}$, and
    $f^{\otimes n}$ is a linear functional such that
    \[
        M^{\otimes n}
        = (A^{\otimes n}\otimes B^{\otimes n} \otimes f^{\otimes n})\ang{(q+2)^n}
    \]
    has rank $q^n$. Applying \Cref{thm:app-oneslice} yields the desired degeneration.
\end{proof}

Combining the above degeneration with \Cref{prop:comp_ex} immediately yields the following asymptotic-rank bound.

\begin{corollary} \label{cor:ar-cw}
    For any $q\geq 2$ and $n\geq 1$, define
    \[ \gamma_{q,n} \coloneq
    \left((q+2)^n + q^{2n/3} - ((q+2)^n - 2(q+1)^n + q^n)^{2/3}\right)^{1/n} \]
    and
    \[ \gamma_q^{(\ref{cor:ar-cw})} \coloneq \min_{n\geq 1} \gamma_{q,n}. \]
    Then the small CW tensor satisfies
    \[ \AR(\cw_q) \leq \gamma_q^{(\ref{cor:ar-cw})} < q+2. \]
\end{corollary}

We let $n^{(\ref{cor:ar-cw})}$ denote the $n$ that reaches the minimum in \Cref{cor:ar-cw}.

We can further apply the second method to iteratively speed up the degeneration. By \Cref{thm:iterate}, we obtain
\begin{align*}
    & \phantom{=} \cw_q^{\otimes 2n} \; \oplus \; 2\odot \cw_q^{\otimes n} \otimes \ang{1, (q+2)^n-2(q+1)^n+q^n,1}\\
    &\oplus \ang{1, [(q+2)^n-2(q+1)^n+q^n]^2 + 2 (q+1)^{2n},1}\\
    &\degen (\ang{(q+2)^n} \oplus \ang{1,q^n,1})^{\otimes 2}.
\end{align*}
Strassen calculus then gives the upper bound
\[ \AR(\cw_q) \leq \max_{\theta\in [2/3,1]} F(\theta) \]
where
\[ F(\theta) \coloneq \sqrt{(r+s^\theta)^2 - (t^2 + 2(q+1)^{2n})^\theta + t^{2\theta}} - t^\theta \]
with $r=(q+2)^n$, $s=q^n$, and $t=r+s - 2(q+1)^n$.
We denote the value obtained by taking $n = n^{(\ref{cor:ar-cw})}$ by ${\gamma_q'}^{(\ref{cor:ar-cw})}$.

The values of $\gamma_q^{(\ref{cor:ar-cw})}, n^{(\ref{cor:ar-cw})}$ and ${\gamma_q'}^{(\ref{cor:ar-cw})}$
for the first few values of $q$ are shown in \Cref{tab:small_cw} (computed via computer-aided evaluation; in fact,
the maximizing $\theta$ in the definition of ${\gamma_q'}^{(\ref{cor:ar-cw})}$ equals $\theta=2/3$).

\begin{table}[htbp]
    \centering
    \begin{tabular}{ccccc}
        \toprule
        $q$ & $\BR(\cw_q)$ & $\gamma_q^{(\ref{cor:ar-cw})}$ & $n^{(\ref{cor:ar-cw})}$
        & ${\gamma_q'}^{(\ref{cor:ar-cw})}$ \\
        \midrule
        $2$ & $4$ & $3.933484$ & $4$ & $3.930872$ \\
        $3$ & $5$ & $4.970152$ & $4$ & $4.967092$ \\
        $4$ & $6$ & $5.986817$ & $5$ & $5.986293$ \\
        $5$ & $7$ & $6.994597$ & $6$ & $6.994516$ \\
        $6$ & $8$ & $7.997665$ & $6$ & $7.997596$ \\
        $7$ & $9$ & $8.998978$ & $7$ & $8.998969$ \\
        $8$ & $10$ & $9.999595$ & $8$ & $9.999593$ \\
        $9$ & $11$ & $10.999850$ & $9$ & $10.999850$ \\
        $10$ & $12$ & $11.999929$ & $9$ & $11.999929$ \\
        \bottomrule
    \end{tabular}
    \caption{Asymptotic rank upper bound for small CW tensors} \label{tab:small_cw}
\end{table}

Note that this degeneration can be iterated indefinitely. We do not pursue further iterations here,
as the resulting numerical improvements are marginal.

We also remark that the following variant of the $\cw_2$ tensor, which first appears in \cite[\S 11]{CW90Laser},
is equally useful from the perspective of the laser method:
\[ \cw_2' \coloneq x_1y_2z_3 + x_1y_3z_2 + x_2y_1z_3 + x_2y_3z_1
+ x_3y_1z_2 + x_3y_2z_1, \]
In particular, if one could prove $\AR(\cw_2') = 3$, then it would follow that $\omega = 2$.
Currently, the best known upper bound is $\Rk(\cw_2') = 4$ (over fields of characteristic $\neq 2$), via
\begin{align*}
	    \cw_2' &= \frac 1 4 [(x_1+x_2+x_3) (y_1+y_2+y_3) (z_1+z_2+z_3)\\
	    & \phantom{\frac 1 4 [} + (-x_1+x_2-x_3) (-y_1+y_2-y_3) (-z_1+z_2-z_3)\\
	    & \phantom{\frac 1 4 [} + (x_1-x_2-x_3) (y_1-y_2-y_3) (z_1-z_2-z_3) \\
	    & \phantom{\frac 1 4 [} + (-x_1-x_2+x_3) (-y_1-y_2+y_3) (-z_1-z_2+z_3) ].
\end{align*}
Let the four rank-one summands in the decomposition above be $a_i b_i c_i$.
Then
\[ a_1b_1 + a_2b_2 - a_3b_3 - a_4b_4  = 4(x_1y_3 + x_3y_1) \]
has rank $2$, matching the same parameter regime as in \Cref{lem:minrk}. Therefore, the preceding analysis
applies to $\cw_2'$ as well.
\begin{corollary}
	    If $\bbF$ has characteristic $\neq 2$, then we also have
	    $\AR(\cw_2') \leq {\gamma_2'}^{(\ref{cor:ar-cw})}$.
\end{corollary}

\begin{remark} \label{rmk:cw_omega}
    If the asymptotic rank of the small CW tensor were optimal, i.e., $\AR(\cw_q)=q+1$, then applying the laser method directly to $\cw_q$ would in fact yield a better upper bound on $\omega$ than the standard analysis based on the big CW tensors $\CW_q$.
    Since we improve the best known upper bounds on $\AR(\cw_q)$, it is natural to ask whether these improvements already translate into a better bound on $\omega$.

    Starting from the degeneration above, the most straightforward way to bound $\omega$ combines the laser method with our new asymptotic rank bounds. Interestingly, we remark that this incurs a ``composition loss'': the asymptotic rank bound transfers some matrix multiplication structure into diagonal tensors, which the laser method cannot exploit.
    A more faithful approach would start from the degeneration used in our proof and combine Strassen calculus with the laser method at that level, and this indeed achieves a slightly better bound than the straightforward approach.

    Nonetheless, our experiments indicate that, with the identities currently produced by our speedup theorems, this hybrid approach still performs worse than Coppersmith--Winograd's analysis of the laser method applied to the first power $\CW_q^{\otimes 1}$.
\end{remark}

\subsection{General tensors} \label{subsec:exponent}

A nontrivial general upper bound on the asymptotic rank of tensors, due to Strassen (implicit in \cite[Proposition~3.6]{Str88AsymSpec}), states that every $d\times d\times d$ tensor satisfies $\AR(T) \leq d^{\frac 2 3 \omega}$.
In this subsection, we show how to slightly improve this bound.

To explain our improvement, we first recall Strassen's $\frac 2 3 \omega$ bound. The following straightforward rewriting shows the restriction $T\leq \ang{d,1,d}$:
\[ T = \sum_{i,j,k=1}^d a_{ijk} x_i y_j z_k =
\sum_{i,j=1}^d x_i y_j \left(\sum_{k=1}^d a_{ijk} z_k\right)
\leq \sum_{i,j,k=1}^d x_i y_j w_{ij} = \ang{d,1,d}. \]
Applying the same argument to the other modes yields $T\leq \ang{d,d,1}$ and $T\leq \ang{1,d,d}$.
Tensoring these three restrictions gives $T^{\otimes 3} \leq \ang{d^2,d^2,d^2}$, and hence $\AR(T)^3 \leq d^{2\omega}$.

We first present a one-slice speedup of the above argument.

\begin{theorem} \label{thm:general_oneslice}
    When $d \geq 3$, any $d\times d\times d$ tensor $T$ satisfies
    \[ \AR(T) \leq \sqrt{d^{\frac 4 3 \omega} + d^{\frac{2}{3}\omega} - (d^3-d^2)^{\frac{\omega}{3}}}
    < d^{\frac 2 3 \omega}. \]
\end{theorem}

\begin{proof}
    We start with the restriction $T^{\otimes 2} \leq \ang{d,d^2,d}$.

    We wish to apply \Cref{prop:add_slice} to append a slice to both sides of this restriction.
    To do so, we first need a linear functional $f\colon W\to \bbF$ for $S=\ang{d,d^2,d}$ such that
    $({\id} \otimes {\id} \otimes f) S$ has large rank. Write
    \[ S = \sum_{i=1}^{d^2} \sum_{j=1}^d \sum_{k=1}^{d^2} x_{ij} y_{jk} z_{ki}. \]
    Let $f$ be defined by $f(z_{ki}) = \iv{i = k}$. Then
    \[ ({\id} \otimes {\id} \otimes f)S = \sum_{i=1}^{d^2} \sum_{j=1}^d x_{ij} y_{ji} \]
    is diagonal, hence has rank $s=d^3$. Applying \Cref{prop:add_slice}
    and \Cref{prop:freelunch_oneslice}, we obtain the degeneration
    \[ T^{\otimes 2} \oplus \ang{1,d^3-d^2,1}  \degen  \ang{d,d^2,d}\oplus \ang{1,d^2,1}. \]

    Passing to the asymptotic spectrum, this means
    \[ \phi(T)^2 \leq d^{\theta_1 + 2\theta_2 + \theta_3} + d^{2\theta_2} - (d^3-d^2)^{\theta_2}. \]

    Writing $\varpi \coloneq \theta_1 + \theta_2 + \theta_3$, by permuting modes we obtain analogous inequalities, we have
    \[ \phi(T)^2 \leq d^{\varpi + \theta_i} + d^{2\theta_i} - (d^3 - d^2)^{\theta_i}. \]
    Since $\theta_1 + \theta_2 + \theta_3 = \varpi$ and $\theta_1,\theta_2,\theta_3\geq 0$,
    there exists $i$ such that $\theta_i \leq \varpi / 3$. One can show that the function
    \[ F(\theta) = d^{\varpi + \theta} + d^{2\theta} - (d^3 - d^2)^{\theta}, \quad \theta\in [0, \varpi/3] \]
    is maximized at $\theta = \varpi / 3$ (the proof is deferred to \Cref{lem:varpi}).
    Furthermore, by \Cref{lem:f_varpi_increasing}, the quantity $F(\varpi/3)$ is increasing in
    $\varpi$ on $[2,\omega]$. Since $\varpi\le \omega$ for all $(\theta_1,\theta_2,\theta_3)\in\specMM$,
    we obtain
    \[ \phi(T)^2 \leq d^{\frac 4 3 \omega} + d^{\frac 2 3 \omega} - (d^3-d^2)^{\omega/3}. \]
    The claim then follows from Strassen duality.
\end{proof}

By appending direct sums rather than a single slice, we obtain a stronger bound.

\begin{lemma}
    Let $\{d_\alpha\}_{1\leq \alpha\leq p}$ be positive integers such that
    $\sum_\alpha d_\alpha \leq d$. Then we have degeneration
    \[ T^{\otimes 2} \; \oplus \; \bigoplus_{\alpha=1}^p \ang{1, (d_\alpha - 1)d^2, 1}
    \; \degen \; \ang{d,d^2,d} \; \oplus \; p\odot \ang{1,d^2,1}. \]
\end{lemma}
\begin{proof}
    We start from the redundant restriction $T^{\otimes 2} \leq \ang{d,d^2,d} \oplus p\odot \ang{1,d^2,1}$.
    Write
    \[ S \coloneq \ang{d,d^2,d}\oplus p\odot\ang{1,d^2,1} =
    \sum_{i=1}^d \sum_{j=1}^{d^2} \sum_{k=1}^d u_{ij} v_{jk} w_{ki}
    - \sum_{\alpha=1}^p \sum_{\ell=1}^{d^2} u'_{\alpha \ell} v'_{\alpha \ell} w'_\alpha, \]
    where the minus sign is introduced to simplify later cancellations.
    Fix a partition $[\sum_\alpha d_\alpha] = I_1\sqcup \cdots \sqcup I_p$ where $|I_\alpha| = d_\alpha$.
    Define the linear map $C$ by augmenting the corresponding linear map of $T^{\otimes 2}\leq \ang{d,d^2,d}$
    by zero on new components, and $C'$ by
    \[ C': \quad w_{ki} \mapsto \begin{cases}
        w'_\alpha & i = k \in I_\alpha,\\
        0 & \text{otherwise},
    \end{cases} \qquad w'_\alpha \mapsto w'_\alpha, \]
    and choose linear maps $A, B$ so that
    \begin{align*}
        0 &= (A\otimes B\otimes C') S\\
        &= \sum_{\alpha=1}^p \left( \sum_{i\in I_\alpha} \sum_{j=1}^{d^2}
        A(u_{ij}) B(v_{ji}) - \sum_{\ell=1}^{d^2} A(u'_{\alpha \ell})
        B(v'_{\alpha \ell}) \right) w'_\alpha,
    \end{align*}
    where the action on the $u$- and $v$-variables is already determined (since it stems from $T^{\otimes 2} \leq \ang{d,d^2,d}$).
    Since $T^{\otimes 2}$ has dimension $d^2$ in each mode, we can choose $A(u'_{\alpha \ell})$ and $B(v'_{\alpha \ell})$ so that the terms cancel.

    Applying \Cref{thm:freelunch_speedup} yields a direct summand $T'=(A'\otimes B'\otimes C')S$
    on the left-hand side. We first expand the form of $({\id} \otimes {\id} \otimes C')S$:
    \[ ({\id} \otimes {\id} \otimes C')S = \sum_{\alpha=1}^p \underbracket{\left( \sum_{i\in I_\alpha} \sum_{j=1}^{d^2}
        u_{ij} v_{ji} - \sum_{\ell=1}^{d^2} u'_{\alpha \ell}
        v'_{\alpha \ell} \right) w'_\alpha}_{\ang{1,(d_\alpha+1)d^2,1}}. \]
    Then, as in the second proof of \Cref{thm:speedup_via_compression}, after applying the contraction
    $A'\otimes B'$, each summand $\ang{1,(d_\alpha+1)d^2,1}$ restricts to
    $\ang{1,(d_\alpha-1)d^2,1}$. This completes the proof.
\end{proof}

Compared with \Cref{thm:general_oneslice}, this yields an improved bound in which the subtractive term is asymptotically stronger.

\begin{theorem}
    When $d \geq 3$, any $d\times d\times d$ tensor $T$ satisfies
    \[ \AR(T) \leq \sqrt{d^{\frac 4 3 \omega} -
    \left\lfloor \frac d 3 \right\rfloor\cdot (2^{\frac \omega 3} - 1)d^{\frac 2 3 \omega}}. \]
\end{theorem}
\begin{proof}
    Apply the previous lemma with $p = \lfloor d/3\rfloor$ and $d_\alpha = 3$ for all $\alpha$.
    It gives the degeneration
    \[ T^{\otimes 2} \oplus p \odot \ang{1,2d^2,1} \degen \ang{d,d^2,d} \oplus p\odot \ang{1,d^2,1}. \]
    Fix an arbitrary spectral point $\phi\in\calX$, and let $(\theta_1,\theta_2,\theta_3)\in\specMM$ be its parameters.
    Passing to the asymptotic spectrum yields
    \[ \phi(T)^2 \le d^{\theta_1+2\theta_2+\theta_3} - (2^{\theta_2}-1)p\,d^{2\theta_2}. \]
    By permuting tensor modes, we obtain the analogous bounds for $i\in\{1,2,3\}$:
    \[ \phi(T)^2 \le d^{\varpi+\theta_i} - (2^{\theta_i}-1)p\,d^{2\theta_i}, \qquad \text{where }\varpi\coloneq \theta_1+\theta_2+\theta_3. \]
    Let $\theta\coloneq \min\{\theta_1,\theta_2,\theta_3\}$, so $\theta\le \varpi/3$.
    Define
    \[ h_\varpi(x) \coloneq d^{\varpi+x} - (2^x-1)p\,d^{2x}. \]
    Then $\phi(T)^2 \le h_\varpi(\theta)$. By \Cref{lem:h_increasing}, $h_\varpi$ is increasing on $[0,\varpi/3]$, hence
    \[ \phi(T)^2 \le h_\varpi(\varpi/3) = d^{4\varpi/3} - (2^{\varpi/3}-1)p\,d^{2\varpi/3}. \]
    By \Cref{lem:g_increasing}, the right-hand side is increasing in $\varpi$ on $[2,\omega]$; since $\varpi\le \omega$ for every $(\theta_1,\theta_2,\theta_3)\in\specMM$, we obtain
    \[ \phi(T)^2 \le d^{4\omega/3} - (2^{\omega/3}-1)p\,d^{2\omega/3}. \]
    Finally, taking the maximum over $\phi\in\calX$ and using Strassen duality yields
    \[ \AR(T)^2 \le d^{4\omega/3} - (2^{\omega/3}-1)p\,d^{2\omega/3}, \]
    i.e.,
    \[ \AR(T) \le \sqrt{d^{\frac 4 3 \omega} -
    \left\lfloor \frac d 3 \right\rfloor\cdot (2^{\frac \omega 3} - 1)d^{\frac 2 3 \omega}}. \qedhere \]
\end{proof}

Let $\sigma(d)$ denote the exponent of $d\times d\times d$ tensors \cite{KM25ARC}, i.e., the smallest $\sigma$ such that every $d\times d\times d$ tensor $T$ satisfies $\AR(T) \leq d^{\sigma}$.
Then the above results imply the following.

\begin{corollary}
    For any $d > 0$, $\sigma(d) < \frac 23 \omega$.
\end{corollary}

\begin{remark} \label{rmk:reduce_mm}
    This answers an open problem posed by Kaski \cite[Open Problem~1]{Kaski2025Talk}, which asks whether Strassen's bound $\sigma(d)\le \frac23\omega$ can be improved. It is worth emphasizing, however, that our improvement has the form $\sigma(d)\le \frac23\omega-\varepsilon_d$, where $\varepsilon_d\to 0$ rapidly as $d\to\infty$.

    In particular, to make progress on fine-grained questions such as the set cover conjecture \cite{bjorklund2024asymptotic,Pratt24SCC}, one would need to achieve $\sigma(d)\le 1.08$.
    While in the most ideal case $\omega = 2$, currently we only have an exponent bound $\sigma(d) \leq 4/3 - o(1)$.
    On the other hand, this current approach also seems ill-suited for improving $\omega$: the strategy is to reduce a generic tensor $T$ to matrix multiplication, but when $T$ is itself a matrix multiplication tensor, this reduction becomes too indirect to be effective.
\end{remark}

To make progress toward the (extended; see \cite[Conjecture~2]{KM25ARC}) asymptotic rank conjecture, we propose the following question, which seems to require ideas beyond a naive application of the techniques developed in this paper.
\begin{problem}
    Prove that there exists $\delta > 0$ such that
    $\sigma(d) \leq \frac 2 3 \omega - \delta$ for all $d$.
\end{problem}

\subsection{A direct sum identity} \label{subsec:directsum}

In this subsection, we use the speedup theorems developed above to derive an intriguing
direct sum identity that generalizes Sch\"onhage's classical construction
\cite[Lemma~6.1]{Sch81Tau}:
\[
    \ang{n,1,m} \oplus \ang{1,(n-1)(m-1),1} \degen \ang{nm+1}.
\]

\begin{theorem} \label{thm:direct_sum_identity}
    Let $p,q$ be positive integers, and let
    $\{n_\alpha\}_{1\le \alpha\le p}$ and $\{m_\beta\}_{1\le \beta\le q}$ be sequences of
    positive integers. Set
    \[
        n=\sum_{\alpha=1}^p n_\alpha,
        \qquad
        m=\sum_{\beta=1}^q m_\beta.
    \]
    Then there exists a degeneration between direct sums of matrix multiplication tensors:
    \[
        \ang{n,1,m}
        \oplus
        \bigoplus_{\substack{1\le \alpha\le p\\1\le \beta\le q}}
        \ang{1,(n_\alpha-1)(m_\beta-1),1}
        \ \degen\
        \ang{nm} \oplus \ang{p,1,q}.
    \]
\end{theorem}

Before proving \Cref{thm:direct_sum_identity}, we briefly compare it with Sch\"onhage's original direct
sum identity. Taking the direct sum over all $(\alpha,\beta)$ gives
\[
    \bigoplus_{\alpha,\beta} \ang{n_\alpha,1,m_\beta}
    \oplus
    \bigoplus_{\alpha,\beta} \ang{1,(n_\alpha-1)(m_\beta-1),1}
    \degen
    \ang{nm} \oplus \ang{pq}.
\]
Compared with \Cref{thm:direct_sum_identity}, the direct-sum tensor
$\bigoplus_{\alpha,\beta} \ang{n_\alpha,1,m_\beta}$ on the left-hand side is more valuable than
$\ang{n,1,m}$. However, this also comes with a bigger cost on the right-hand side: the slack term
$\ang{pq}$ replaces the one-slice tensor $\ang{p,1,q}$.

In the special case $n_\alpha=n$ and $m_\beta=m$ for all $\alpha,\beta$, our identity becomes
\[
    \ang{pn,1,qm} \oplus pq\odot \ang{1,(n-1)(m-1),1}
    \degen
    \ang{pqnm} \oplus \ang{p,1,q}.
\]
If instead one tensors Sch\"onhage's identity for $(n,m)$ with $\ang{p,1,q}$, one obtains
\[
    \ang{pn,1,qm} \oplus \ang{p,(n-1)(m-1),q}
    \degen
    nm\odot \ang{p,1,q} \oplus \ang{p,1,q}.
\]
Here the direct summand $nm \odot \ang{p,1,q}$ on the right-hand side is cheaper than $\ang{pqnm}$,
but the tensor $\ang{p,(n-1)(m-1),q}$ on the left-hand side is also less valuable than
$pq\odot \ang{1,(n-1)(m-1),1}$.

In sum, our new identity does not seem to be easily derived from Sch\"onhage's identity alone.

\begin{proof}[Proof of \Cref{thm:direct_sum_identity}]
    We start from a redundant restriction
    \[
        \ang{n,1,m} \le \ang{nm} \oplus \ang{p,1,q}.
    \]
    Write
    \[
        \ang{n,1,m}
        = \sum_{i=1}^n \sum_{j=1}^m x_i y_j z_{ij},
        \qquad
        S \coloneq \ang{nm} \oplus \ang{p,1,q}
        = \sum_{i=1}^n \sum_{j=1}^m u_{ij} v_{ij} w_{ij}
        - \sum_{\alpha=1}^p \sum_{\beta=1}^q u'_\alpha v'_\beta w'_{\alpha\beta},
    \]
    where the minus sign is introduced to simplify later cancellations.

    Fix partitions $[n]=I_1\sqcup\cdots\sqcup I_p$ and
    $[m]=J_1\sqcup\cdots\sqcup J_q$ with $|I_\alpha|=n_\alpha$ and $|J_\beta|=m_\beta$.
    Define linear maps $A,B,C$ by
    \begin{align*}
        A:&\quad x_i \mapsfrom u_{ij}, \qquad \sum_{i\in I_\alpha} x_i \mapsfrom u'_\alpha,\\
        B:&\quad y_j \mapsfrom v_{ij}, \qquad \sum_{j\in J_\beta} y_j \mapsfrom v'_\beta,\\
        C:&\quad z_{ij} \mapsfrom w_{ij}, \qquad 0 \mapsfrom w'_{\alpha\beta}.
    \end{align*}
    These maps define a restriction $\ang{n,1,m} \le S$.

    To apply the free-lunch speedup theorem (\Cref{thm:freelunch_speedup}), define a subspace
    $C'\subseteq W^\vee$ via the linear map
    \[
        C':\quad
        w'_{\alpha(i)\beta(j)} \mapsfrom w_{ij},
        \qquad
        w'_{\alpha\beta} \mapsfrom w'_{\alpha\beta},
    \]
    where $\alpha(i)$ (resp.~$\beta(j)$) is the unique index such that
    $i\in I_{\alpha(i)}$ (resp.~$j\in J_{\beta(j)}$).
    A direct computation shows that $(A\otimes B\otimes C')S=0$.

    Applying \Cref{thm:freelunch_speedup} yields a direct summand
    $T'=(A'\otimes B'\otimes C')S$ on the left-hand side.

    Let us make the resulting subspaces $A'\subseteq U^\vee$ and $B'\subseteq V^\vee$ explicit.
    By definition,
    \[
        A' = \{u\in U^\vee : (u\otimes B\otimes C')S=0\},
        \qquad
        B' = \{v\in V^\vee : (A\otimes v\otimes C')S=0\}.
    \]
    Unwinding the maps $A,B$ and the definition of $C'$, one checks that $A'$ is the annihilator
    of the subspace spanned by the vectors
    \[
        u'_\alpha - \sum_{i\in I_\alpha} u_{ij}
        \quad (\text{for all } \alpha\in[p],\ j\in[m]),
    \]
    and similarly $B'$ is the annihilator of the subspace spanned by
    \[
        v'_\beta - \sum_{j\in J_\beta} v_{ij}
        \quad (\text{for all } \beta\in[q],\ i\in[n]).
    \]
    In particular, in $T'$ the relations
    \[
        u'_\alpha = \sum_{i\in I_\alpha} u_{ij}\quad (\text{for all } \alpha\in[p],\ j\in[m])
        \qquad\text{and}\qquad
        v'_\beta = \sum_{j\in J_\beta} v_{ij}\quad (\text{for all } \beta\in[q],\ i\in[n])
    \]
    hold simultaneously.
    For each $\alpha$ choose a distinguished index $i_\alpha^*\in I_\alpha$ and for each $\beta$
    choose a distinguished index $j_\beta^*\in J_\beta$.
    These relations allow us to eliminate the variables $u_{i_\alpha^* j}$ and $v_{i j_\beta^*}$ via
    \[
        u_{i_\alpha^* j}
        = u'_\alpha - \sum_{i\in I_\alpha\setminus\{i_\alpha^*\}} u_{ij},
        \qquad
        v_{i j_\beta^*}
        = v'_\beta - \sum_{j\in J_\beta\setminus\{j_\beta^*\}} v_{ij}.
    \]

    We now impose a further restriction by applying suitable projections on the first two modes.
    Define linear maps $\pi_U$ and $\pi_V$ by
    \[
        \pi_U(u'_\alpha)=0,\qquad \pi_U(u_{i\,j_\beta^*})=0\ (\forall i,\beta),\qquad \pi_U(u_{ij})=u_{ij}\ (j\neq j_{\beta(j)}^*),
    \]
    and
    \[
        \pi_V(v'_\beta)=0,\qquad \pi_V(v_{i_\alpha^*\,j})=0\ (\forall \alpha,j),\qquad \pi_V(v_{ij})=v_{ij}\ (i\neq i_{\alpha(i)}^*).
    \]
    Applying $\pi_U$ and $\pi_V$ to the first two modes of $T'$ yields a further restriction
    \[ T^* \coloneq (\pi_U\otimes \pi_V\otimes \id)T' \le T'. \]

    By construction, $\pi_U$ kills all $u'_\alpha$ and all variables in the distinguished columns
    $\{u_{i\,j_\beta^*}\}$, and $\pi_V$ kills all $v'_\beta$ and all variables in the distinguished rows
    $\{v_{i_\alpha^*\,j}\}$.
    After eliminating $u_{i_\alpha^* j}$ and $v_{i j_\beta^*}$ using the identities above, every term
    that contains an eliminated variable involves either $u'_\alpha$ (or a distinguished column) or
    $v'_\beta$ (or a distinguished row), and hence is annihilated by $\pi_U\otimes\pi_V$.
    Therefore, the only surviving terms are those with
    $i\in I_\alpha\setminus\{i_\alpha^*\}$ and $j\in J_\beta\setminus\{j_\beta^*\}$.

    The resulting tensor is
    \[
        T^*
        = \sum_{\alpha=1}^p \sum_{\beta=1}^q
        \left(
            \sum_{i\in I_\alpha\setminus\{i_\alpha^*\}}
            \sum_{j\in J_\beta\setminus\{j_\beta^*\}}
            u_{ij} v_{ij}
        \right) w_{\alpha\beta}
        \ \cong\
        \bigoplus_{\alpha,\beta} \ang{1,(n_\alpha-1)(m_\beta-1),1}.
    \]
    Composing the restrictions and degenerations yields the claimed identity.
\end{proof}

\bibliography{main}

\appendix
\crefalias{section}{appendix}
\crefalias{subsection}{subappendix}

\section{Degenerations between direct sums of matrix multiplication} \label{sec:omega}

The goal of this appendix is to clarify how the direct-sum degeneration used by
Sch\"onhage \cite[\S 8]{Sch81Tau},
\begin{equation}
    \ang{1,5,20} \oplus \ang{10,2,5} \oplus \ang{10,10,1}
    \degen \ang{132} \oplus 6 \odot \ang{1,1,2}, \label{eqn:schonhage_degen}
\end{equation}
can be analyzed (via Strassen calculus) to obtain the inequality
\begin{equation}
    3 \cdot 100^{\omega/3}
    \leq 132 + 6 \cdot 2^{\omega/3}. \label{eqn:schonhage_eqn}
\end{equation}
Moreover, this analysis is \emph{optimal} under a suitable formalization.

More generally, given a single degeneration between direct sums of matrix multiplication tensors,
one may ask for the best upper bound on $\omega$ that can be extracted from it:
\begin{equation}
    \bigoplus_{i=1}^t \ang{n_i^{(1)}, n_i^{(2)}, n_i^{(3)}} \degen \bigoplus_{j=1}^s \ang{m_j^{(1)}, m_j^{(2)}, m_j^{(3)}}.
    \label{eqn:disjoint_mm_degen}
\end{equation}

We begin with a Strassen-calculus proof of the classical asymptotic sum inequality, which yields an upper
bound in the special case where the right-hand side is diagonal, i.e., $m_j^{(1)}=m_j^{(2)}=m_j^{(3)}=1$.
This proof is due to Strassen \cite[\S 4]{Str88AsymSpec}.
\begin{theorem}[Asymptotic sum inequality {\cite[Theorem~7.1]{Sch81Tau}}]
    If there exists a degeneration,
    \[ \bigoplus_{i=1}^t \ang{n_i^{(1)}, n_i^{(2)}, n_i^{(3)}} \degen \ang{s}, \]
    let $\tau$ be the (unique) solution to the equation
    \[ \sum_{i=1}^t (n_i^{(1)} n_i^{(2)} n_i^{(3)})^\tau = s, \]
    then $\omega \leq 3\tau$.
\end{theorem}
\begin{proof}
    \emph{Step 1: Pass to the asymptotic spectrum.}
    Fix $\phi\in \calX$ with parameters $(\theta_1,\theta_2,\theta_3)\in\specMM$. Applying $\phi$ to the degeneration yields
    \[
        F(\theta_1, \theta_2, \theta_3)
        \coloneq
        \sum_{i=1}^t (n_i^{(1)})^{\theta_1}(n_i^{(2)})^{\theta_2}(n_i^{(3)})^{\theta_3}
        \le s.
    \]

    \emph{Step 2: Symmetry of tensors.} We can permute the modes
    and obtain degenerations that permute the parameters of matrix multiplication. For any $\sigma\in \Sym_3$, we also have a degeneration
    \[ \bigoplus_{i=1}^t \ang{n_i^{(\sigma^{-1}(1))}, n_i^{(\sigma^{-1}(2))}, n_i^{(\sigma^{-1}(3))}} \degen \ang{s}, \]
    hence $F(\theta_{\sigma(1)},\theta_{\sigma(2)},\theta_{\sigma(3)}) \leq s$.

    \emph{Step 3: Average by convexity.} The function $F$ is convex (it is a sum of exponentials), so by Jensen's inequality,
    \[
        F\!\Bigl(\bbE_\sigma[(\theta_{\sigma(1)},\theta_{\sigma(2)},\theta_{\sigma(3)})]\Bigr)
        \le \bbE_\sigma[F(\theta_{\sigma(1)},\theta_{\sigma(2)},\theta_{\sigma(3)})]
        \le s.
    \]
    Let $\varpi \coloneq \theta_1 + \theta_2 + \theta_3$. Averaging over $\sigma\in\Sym_3$ gives
    \[
        \bbE_\sigma[(\theta_{\sigma(1)},\theta_{\sigma(2)},\theta_{\sigma(3)})]
        = (\varpi/3,\varpi/3,\varpi/3).
    \]
    Since $\varpi\le \omega$ for all $(\theta_1,\theta_2,\theta_3)\in\specMM$ and $F$ is coordinatewise nondecreasing, we obtain
    $F(\omega/3, \omega/3, \omega/3) \leq s$, which implies the claim.
\end{proof}

We emphasize that the first two steps in the proof above amount to a simple geometric relaxation, which
extends verbatim to the case where the right-hand side is also a direct sum of matrix multiplication tensors:

\emph{Step 1: Pass to the asymptotic spectrum.} Every point $(\theta_1,\theta_2,\theta_3)\in \specMM$ must lie in the feasible region
\[ \specMM \subseteq \Delta_{\eqref{eqn:disjoint_mm_degen}} \coloneq \left\{ (\theta_1, \theta_2, \theta_3) : \sum_{i=1}^t (n_i^{(1)})^{\theta_1}
(n_i^{(2)})^{\theta_2}(n_i^{(3)})^{\theta_3}
\leq \sum_{j=1}^s
(m_j^{(1)})^{\theta_1}
(m_j^{(2)})^{\theta_2}
(m_j^{(3)})^{\theta_3} \right\}. \]

\emph{Step 2: Symmetry.} Permuting tensor modes yields the containment
\[ \specMM \subseteq \Delta_{\eqref{eqn:disjoint_mm_degen}}^\sigma
= \{ (\theta_1, \theta_2, \theta_3) : (\theta_{\sigma(1)}, \theta_{\sigma(2)}, \theta_{\sigma(3)}) \in \Delta_{\eqref{eqn:disjoint_mm_degen}} \} \]
for all $\sigma\in \Sym_3$.

\emph{Step 3: Optimize over the relaxation.} Intersecting all six permuted regions gives the strongest symmetric relaxation:
$\specMM \subseteq \Delta_{\eqref{eqn:disjoint_mm_degen}}^\Sym \coloneq \bigcap_{\sigma\in\Sym_3} \Delta_{\eqref{eqn:disjoint_mm_degen}}^\sigma$.
By Strassen duality, we then obtain an upper bound on~$\omega$:
\[ \omega = \max_{\specMM} (\theta_1+\theta_2+\theta_3)
\leq \max_{\Delta_{\eqref{eqn:disjoint_mm_degen}}^\Sym} (\theta_1+\theta_2+\theta_3). \]

From this viewpoint, the classical asymptotic sum inequality says that, in the diagonal case,
the region $\Delta_{\eqref{eqn:disjoint_mm_degen}}$ is convex, so the maximum of $\theta_1+\theta_2+\theta_3$ over
$\Delta_{\eqref{eqn:disjoint_mm_degen}}^\Sym$ is attained on the principal diagonal $\theta_1=\theta_2=\theta_3$, which simplifies the analysis.

Therefore, the analysis of the direct-sum identity \eqref{eqn:schonhage_degen} can be carried out via
computer-aided optimization. The inequality induced by the degeneration is
\[ \Delta_{\eqref{eqn:schonhage_degen}} \coloneq \left\{ (\theta_1,\theta_2,\theta_3) :
5^{\theta_2} \cdot 20^{\theta_3} + 10^{\theta_1} \cdot 2^{\theta_2} \cdot 5^{\theta_3}
+ 10^{\theta_1} \cdot 10^{\theta_2} \leq 132 + 6 \cdot 2^{\theta_3}\right\}. \]

It turns out (and can be verified in standard mathematical software packages) that the maximum of
$\theta_1+\theta_2+\theta_3$ over $\Delta_{\eqref{eqn:schonhage_degen}}^\Sym$ is attained at
$\theta_1=\theta_2=\theta_3$. This justifies \eqref{eqn:schonhage_eqn} as a valid upper bound on~$\omega$.

One may then ask whether an asymptotic-sum-inequality statement holds in \emph{full generality}. Namely, can one conclude
\begin{equation}
    \eqref{eqn:disjoint_mm_degen}\implies \sum_{i=1}^t (n_i^{(1)} n_i^{(2)} n_i^{(3)})^{\omega / 3} \leq \sum_{j=1}^s (m_j^{(1)} m_j^{(2)} m_j^{(3)})^{\omega / 3}.
    \label{eqn:general_asym_sum_ineq}
\end{equation}

The answer is nuanced: in a certain relativized sense the answer is \emph{no}, but it is also plausible that the answer is \emph{yes} for tensors.

We first explain the \emph{no} side. A natural first attempt is to let $\tau$ be the largest real satisfying
\[ \sum_{i=1}^t (n_i^{(1)} n_i^{(2)} n_i^{(3)})^{\tau} \leq \sum_{j=1}^s (m_j^{(1)} m_j^{(2)} m_j^{(3)})^{\tau}, \]
Then $(\tau,\tau,\tau)\in \Delta_{\eqref{eqn:disjoint_mm_degen}}^\Sym$, so $3\tau$ provides a lower bound on
the best upper bound on $\omega$ obtainable from the identity. However, the converse need not hold:
the maximum of $\theta_1+\theta_2+\theta_3$ over $\Delta_{\eqref{eqn:disjoint_mm_degen}}^\Sym$ may be attained away from
the principal diagonal $\theta_1=\theta_2=\theta_3$. In such cases, the optimal bound on $\omega$ extracted from the degeneration can be strictly larger than $3\tau$.

A cautious reader may notice that we have been deliberately vague about what it means to obtain ``the best bound on $\omega$ from a single identity''.
In particular, one should clarify what is meant by ``using only the identity'' as opposed to using additional tensor-specific facts.

The underlying intuition is that the classical asymptotic sum inequality---and the bound \eqref{eqn:schonhage_eqn} derived from \eqref{eqn:schonhage_degen}---is, in a sense, \emph{blind} to the arithmetic structure of tensors.
Once one abstracts the basic axioms governing rank and degeneration, the same argument goes through without using any additional tensor-specific structure.
From the perspective of complexity theory, this can be phrased as a form of \emph{relativization} \cite{BGS75Relativization}:
in any mathematical structure that satisfies the same axioms and ``knows'' only the single identity \eqref{eqn:disjoint_mm_degen}, an asymptotic-sum-inequality type bound still holds.
In particular, one can show that there exists a relativized world in which $\omega$ cannot be smaller than
$\max_{\Delta_{\eqref{eqn:disjoint_mm_degen}}^{\Sym}}(\theta_1+\theta_2+\theta_3)$.
One way to formalize this is to work in a \emph{universal commutative semiring} that encodes the axioms together with the given identity.
We give a rigorous treatment in the subsequent subsection.

We briefly indicate why the answer \emph{might} nevertheless be yes. The discussion above only shows that
\eqref{eqn:general_asym_sum_ineq} is too strong to hold in all relativized settings; additional structure specific to tensors might still imply it.
A natural candidate is \emph{Strassen's convexity conjecture}, which asserts that $\specMM$ is a convex subset of~$\bbR^3$.
If $\specMM$ were convex, then \eqref{eqn:general_asym_sum_ineq} would follow almost immediately. Moreover, the conjecture is not entirely speculative:
Strassen proved partial results showing that $\specMM$ is star-convex with respect to the points $(0,1,1)$, $(1,0,1)$, and $(1,1,0)$; see~\cite[\S 5]{Str88AsymSpec}.
For further discussion, see \cite[\S 4, 5]{Str88AsymSpec} or the more detailed exposition in \cite[Part II, III]{WZ25Spectra}.

\subsection{Strassen duality and asymptotic sum inequality in general rings}

To make the preceding discussion precise, we briefly recall the general notion of rank and Strassen duality in the abstract setting of commutative semirings.
This framework is explained in detail by Wigderson and Zuiddam \cite[Part I]{WZ25Spectra}.

\myparagraph{General framework of Strassen duality.}

Let $\calR$ be a commutative semiring with additive identity $0$ and multiplicative identity $1$.
We think of elements of $\calR$ as the ``objects'' of interest. For example, in the $\bbF$-tensor world,
$\calR = \calR_\bbF$ can be taken to be the set of all $\bbF$-tensors (up to isomorphism), with addition given by direct sum $\oplus$
and multiplication given by tensor product $\otimes$; the zero tensor plays the role of $0$, and the unit diagonal tensor $\ang{1}$ plays the role of $1$.

We also equip $\calR$ with a preorder $P$ (written $\le_P$), which encodes the basic comparison relation between objects (e.g., restriction or degeneration).
In this setting, the rank of an element $a\in\calR$ is defined as
\[
    \Rk_P(a)\coloneq \min\{n\in\bbN : a\le_P n\},
\]
where $n\in\calR$ denotes the sum of $n$ copies of the unit~$1$.
To ensure that this notion behaves as expected, one assumes the following axioms; a preorder satisfying them will be called a \emph{Strassen preorder}:
\begin{description}
    \axitem{sp-preorder}{(Preorder)} Reflexivity ($a\le_P a$) and transitivity ($a\le_P b, b\le_P c \implies a\le_P c$).
    \axitem{sp-compatibility}{(Compatibility)} $0\le_P a$, and the preorder respects addition and multiplication:
    $a\le_P b \implies a+c\le_P b+c$ and $ac\le_P bc$.
    \axitem{sp-faithfulness}{(Faithfulness on natural numbers)} For $n,m\in\bbN$, one has $n\le_P m$ if and only if $n\le m$ as natural numbers.
    \axitem{sp-archimedean}{(Strong Archimedean property)} For any $a\neq 0$, there exists $n\in\bbN$ such that $1\le_P a\le_P n$.
\end{description}
From now on, we omit $P$ when unambiguous.
Under these axioms, the rank function is subadditive and submultiplicative, and one can define the asymptotic rank $\AR(a)$ in the usual way.
For example, in the tensor setting, both restriction $\le$ and degeneration $\degen$ satisfy these axioms.

One can then define the asymptotic spectrum $\calX$ of $(\calR,\le)$: it consists of all functions $\phi\colon \calR\to \bbR_{\geq 0}$ that
\begin{description}
    \axitem{spectrum-hom}{(Semiring homomorphism)} $\phi(a+b) = \phi(a) + \phi(b)$ and $\phi(ab) = \phi(a)\phi(b)$.
    \axitem{spectrum-normalization}{(Normalization)} $\phi(n) = n$ for all $n\in\bbZ_{\ge 0}$.
    \axitem{spectrum-monotonicity}{(Monotonicity)} $a\le b \implies \phi(a)\le \phi(b)$.
\end{description}

Strassen duality (as discussed in \Cref{sec:str_calc}) remains valid in this abstract setting. For our purposes, we mainly use the identity
\[ \AR(a) = \max_{\phi\in\calX} \phi(a). \]

\myparagraph{Relativization of matrix multiplication.}
We now formulate an abstract semiring \emph{representing} matrix multiplication, in order to study the asymptotic sum inequality.
We define $\RMM$ to be the commutative semiring of finite formal $\bbN$-linear combinations of symbols $\ang{n,m,p}$:
\[ \RMM = \left\{ \sum_{n,m,p} a_{nmp} \ang{n,m,p} : a_{nmp} \in \bbN, \sum_{n,m,p} a_{nmp} < \infty \right\}. \]

Addition in $\RMM$ is defined by adding coefficients coordinatewise, and multiplication is defined on generators by
$\ang{n,m,p} \cdot \ang{n',m',p'} = \ang{nn', mm', pp'}$ and extended bilinearly:
\[ \left(\sum_{n,m,p} a_{nmp} \ang{n,m,p}\right)
   \left(\sum_{n',m',p'} b_{n'm'p'} \ang{n',m',p'} \right)
   = \sum_{\substack{n,m,p\\ n',m',p'}} a_{nmp}b_{n'm'p'} \ang{nn', mm', pp'}. \]
Note that $\ang{1,1,1}$ is the multiplicative identity, so we will identify $\ang{1,1,1}$ with~$1$.

One may view elements of $\RMM$ as formal analogues of direct sums of matrix multiplication tensors,
e.g., $\sum_{n,m,p} a_{nmp} \ang{n,m,p}$ corresponds heuristically to $\bigoplus_{n,m,p} a_{nmp} \odot \ang{n,m,p}$.
However, elements of $\RMM$ are purely formal: they do not ``know'' anything about tensors beyond the axioms we impose below.

We next define a preorder by specifying a set of axioms. Let $\Sigma$ be a collection of inequalities $a\le b$ with $a,b\in\RMM$,
encoding a minimal set of properties of matrix multiplication. Concretely, $\Sigma$ contains the following families of inequalities (for all $n\in\bbZ_+$):
\begin{description}
    \axitem{sigma-monotonicity}{(Monotonicity)} For any $n\in \bbZ_+$, we have $\ang{n,1,1} \leq \ang{n+1,1,1}$, $\ang{1,n,1} \leq \ang{1,n+1,1}$
    and $\ang{1,1,n} \leq \ang{1,1,n+1}$.
    \axitem{sigma-trivial}{(Trivial upper bound)} For any $n\in \bbZ_+$, we have $\ang{n,1,1}\leq n$, $\ang{1,n,1}\leq n$
    and $\ang{1,1,n}\leq n$.
\end{description}
Given a set of axioms $\Sigma$, we write $\Sigma \vdash a\le b$ if the inequality $a\le b$ can be \emph{derived} from $\Sigma$ in finitely many steps
using only \axref{sp-preorder}{(Preorder)} and \axref{sp-compatibility}{(Compatibility)}.

We define a preorder $P(\Sigma)$ on $\RMM$ by declaring $a\le_{P(\Sigma)} b$ if and only if $\Sigma \vdash a\le b$.

\begin{proposition}
    The preorder $P(\Sigma)$ is a Strassen preorder.
\end{proposition}
\begin{proof}
    By construction, $P(\Sigma)$ is a preorder and is compatible with addition and multiplication.
    We verify the remaining axioms.

    For \axref{sp-archimedean}{(Strong Archimedean property)}, let $a=\sum_{n,m,p} a_{nmp}\ang{n,m,p}\in\RMM$. From $\Sigma$ one can derive
    \[ \sum_{n,m,p} a_{nmp} \leq \sum_{n,m,p} a_{nmp} \ang{n,m,p} \leq \sum_{n,m,p} a_{nmp} \cdot nmp. \]
    If $a\neq 0$, then $\sum_{n,m,p} a_{nmp} \ge 1$, so in particular $1\le a \le N$ for some $N\in\bbN$.

    Finally, for \axref{sp-faithfulness}{(Faithfulness on natural numbers)}, observe that each inequality in $\Sigma$ holds when interpreting elements of $\RMM$ in the tensor semiring $\calR_\bbF$ (with restriction preorder).
    Hence, if $\Sigma \vdash n\le m$, then $\ang{n}\le \ang{m}$ in the tensor world, which forces $n\le m$ as natural numbers.
\end{proof}

Since $P(\Sigma)$ is a Strassen preorder, Strassen duality applies, and it becomes meaningful to describe its asymptotic spectrum.
In particular, the parametrization in \Cref{prop:param-specmm} still holds: re-running the proof shows that it relies only on \axref{sigma-monotonicity}{(Monotonicity)}, which is already contained in $\Sigma$.
Moreover, in this setting the asymptotic spectrum is naturally parametrized by triples $(\theta_1,\theta_2,\theta_3)$, since every element of $\RMM$ is built from generators $\ang{n,m,p}$.

Now fix a collection $\Gamma$ of additional identities of interest (i.e., inequalities $a\le b$ in $\RMM$). In principle, $\Gamma$ may be infinite---for instance, it could encode an abstracted family of degenerations arising from the laser method---but the reader may keep in mind the case where $\Gamma$ consists of a single identity, such as Sch\"onhage's.

We analogously define the preorder $P(\Sigma \cup \Gamma)$ by closing $\Sigma\cup\Gamma$ under derivations.
It is immediate that $P(\Sigma\cup \Gamma)$ satisfies the axioms of a Strassen preorder except, potentially, \axref{sp-faithfulness}{(Faithfulness on natural numbers)}:
for an arbitrary choice of $\Gamma$, one could force a contradiction (e.g., by including $1\le 0$).
In the applications we have in mind, $\Gamma$ comes from genuine tensor identities, and faithfulness holds.

From now on, we assume $P(\Sigma \cup \Gamma)$ is indeed a Strassen preorder, and now try to describe the structure of the corresponding $\specMM$
in this case, denoted as $\specMM(\Gamma)$.

For an element in $\RMM$:
\[ a = \sum_{n,m,p} a_{nmp} \ang{n,m,p}, \]
and for $(\theta_1,\theta_2,\theta_3)$ we write
\[ a(\theta_1, \theta_2, \theta_3) \coloneq \sum_{n,m,p} a_{nmp} n^{\theta_1} m^{\theta_2} p^{\theta_3}. \]

\begin{proposition}
    Let $\Delta_\Gamma$ be the feasible region cut out by the inequalities in $\Gamma$:
    \[
        \Delta_\Gamma \coloneq \Bigl\{(\theta_1, \theta_2, \theta_3) :
        a(\theta_1, \theta_2, \theta_3) \le b(\theta_1,\theta_2,\theta_3) \text{ for all } (a\le b)\in \Gamma \Bigr\}.
    \]
    If $P(\Sigma \cup \Gamma)$ is a Strassen preorder, then $\specMM(\Gamma) = [0,1]^3 \cap \Delta_\Gamma$.
\end{proposition}
\begin{proof}
    The spectrum $\specMM(\Gamma)$, by definition, should be the set of $(\theta_1,\theta_2,\theta_3)$
    that the function defined by
    \[ \phi (a) = a(\theta_1, \theta_2, \theta_3) \]
    satisfies all the axioms of an asymptotic-spectrum point. It is straightforward to verify that $\phi$ is nonnegative and satisfies \axref{spectrum-hom}{(Semiring homomorphism)} and \axref{spectrum-normalization}{(Normalization)}.

    The only nontrivial condition is \axref{spectrum-monotonicity}{(Monotonicity)}. If $\Sigma \cup \Gamma \vdash a\le b$, then we must have
    $a(\theta_1, \theta_2, \theta_3) \le b(\theta_1, \theta_2, \theta_3)$. Since the derivation rules are compatible with passing to nonnegative real numbers,
    it suffices to enforce these inequalities for the generating set $\Sigma\cup\Gamma$.

    The inequalities in $\Gamma$ impose the region $\Delta_\Gamma$, while the axioms in $\Sigma$ impose:
    \begin{itemize}
        \item \axref{sigma-monotonicity}{(Monotonicity)} implies $\theta_i \geq 0$.
        \item \axref{sigma-trivial}{(Trivial upper bound)} implies $\theta_i \leq 1$.
    \end{itemize}
    Thus $\Sigma$ forces $(\theta_1,\theta_2,\theta_3)\in[0,1]^3$, and the claim follows.
\end{proof}

Define
\[
    \omega_\Gamma \coloneq \log_2 \bigl(\AR_{P(\Sigma\cup \Gamma)}(\ang{2,2,2})\bigr).
\]
Since \Cref{prop:specmm-sum-theta} can also be checked to relativize\footnote{Strictly speaking, something breaks down for asymptotic subrank $\AQ$, but we will not use it.}, we obtain the characterization
\[ \omega_\Gamma = \max_{\specMM(\Gamma)} (\theta_1+\theta_2+\theta_3). \]

On the one hand, $\omega_\Gamma$ is the optimum of $\theta_1+\theta_2+\theta_3$ over the feasible region $[0,1]^3 \cap \Delta_\Gamma$.
On the other hand, unwinding the definition yields the following ``derivation-only'' interpretation:
\begin{itemize}
    \item Start from the axioms $\Sigma$ together with the auxiliary identities $\Gamma$.
    \item Derive (in finitely many steps, using only \axref{sp-preorder}{(Preorder)} and \axref{sp-compatibility}{(Compatibility)}) an inequality of the form $\ang{n,n,n}\le r$.
    \item Conclude the bound $\omega \le \log_n r$.
    \item Take the infimum over all such bounds (using $\omega = \inf_{n\ge 1}\log_n \Rk(\ang{n,n,n})$ in this abstract setting).
\end{itemize}
This provides a precise meaning of ``using only the identities in $\Gamma$''.
Under this notion, Strassen calculus is \emph{complete}: the optimum over $[0, 1]^3\cap \Delta_\Gamma$ is not merely an upper bound on $\omega_\Gamma$; it equals~$\omega_\Gamma$.

\myparagraph{Asymptotic sum inequality.}

We are now ready to formulate a relativized version of the asymptotic sum inequality.
There is, however, one remaining ingredient that we have implicitly used but have not yet formalized in $\RMM$:
the symmetry under permuting tensor modes.

Let $a$ be an element
\[ a = \sum_{n^{(1)},n^{(2)},n^{(3)}} a_{n^{(1)}n^{(2)}n^{(3)}} \ang{n^{(1)},n^{(2)},n^{(3)}}, \]
and let $\sigma\in \Sym_3$ be a permutation. We write
\[ a^\sigma \coloneq \sum_{n^{(1)},n^{(2)},n^{(3)}}
a_{n^{(1)}n^{(2)}n^{(3)}} \ang{n^{(\sigma^{-1}(1))},n^{(\sigma^{-1}(2))},n^{(\sigma^{-1}(3))}}. \]

\begin{description}
    \axitem{rmm-sym}{(Symmetry)} For any inequality $a\le b$ and any $\sigma\in \Sym_3$, the permuted inequality $a^\sigma \le b^\sigma$ also holds.
\end{description}

At first glance, \axref{rmm-sym}{(Symmetry)} appears to be an additional derivation rule, which would force us to rebuild the preceding formalism.
This is unnecessary: we can instead close the axiom set under permutations once and for all.

\begin{lemma}
    Let $\Gamma^\sigma \coloneq \{ a^\sigma \leq b^\sigma : (a\leq b) \in \Gamma\}$,
    and let $\Gamma^{\Sym} \coloneq \bigcup_{\sigma\in \Sym_3} \Gamma^\sigma$.
    Then $P(\Sigma \cup \Gamma^{\Sym})$ coincides with the preorder obtained from $\Sigma\cup\Gamma$ by allowing derivations using
    \axref{rmm-sym}{(Symmetry)} in addition to \axref{sp-preorder}{(Preorder)} and \axref{sp-compatibility}{(Compatibility)}.
\end{lemma}
\begin{proof}
    Let $P'$ denote the preorder obtained by closing $\Sigma\cup\Gamma$ under derivations using
    \axref{sp-preorder}{(Preorder)}, \axref{sp-compatibility}{(Compatibility)}, and \axref{rmm-sym}{(Symmetry)}.
    Since $P'$ is closed under \axref{rmm-sym}{(Symmetry)} by definition, it contains every permuted axiom in $\Gamma^{\Sym}$, and hence
    $\Sigma\cup\Gamma^{\Sym}\subseteq P'$. Therefore $P(\Sigma\cup\Gamma^{\Sym})\subseteq P'$.

    Conversely, we claim that $P(\Sigma\cup\Gamma^{\Sym})$ is closed under \axref{rmm-sym}{(Symmetry)}.
    Indeed, if $\Sigma\cup\Gamma^{\Sym}\vdash a\le b$, then there is a derivation of $a\le b$ from axioms in $\Sigma\cup\Gamma^{\Sym}$
    using only \axref{sp-preorder}{(Preorder)} and \axref{sp-compatibility}{(Compatibility)}. Applying a fixed permutation $\sigma\in\Sym_3$
    to each line of the derivation yields a valid derivation of $a^\sigma \le b^\sigma$, because $\Sigma\cup\Gamma^{\Sym}$ is stable under permutation
    and the preorder/compatibility rules commute with permuting modes.
    Thus $P(\Sigma\cup\Gamma^{\Sym})$ is closed under \axref{rmm-sym}{(Symmetry)}, so $P'\subseteq P(\Sigma\cup\Gamma^{\Sym})$.

    Combining the inclusions gives $P' = P(\Sigma\cup\Gamma^{\Sym})$.
\end{proof}

Putting the preceding formalization together, we obtain a clean ``complete'' statement for the generalized asymptotic sum inequality in this relativized setting.
\begin{proposition}[Generalized asymptotic sum inequality]
    Let $\Gamma$ be a collection of inequalities in $\RMM$.
    Allowing derivations using \axref{sp-preorder}{(Preorder)}, \axref{sp-compatibility}{(Compatibility)}, and \axref{rmm-sym}{(Symmetry)},
    the optimal bound on $\omega$ obtainable from $\Gamma$ equals $\omega_{\Gamma^{\Sym}}$. Equivalently,
    it is the maximum of $\theta_1+\theta_2+\theta_3$ over the feasible region $[0,1]^3 \cap \Delta_{\Gamma^{\Sym}}$.
\end{proposition}

\section{Analysis facts}

\begin{lemma} \label{lem:varpi}
    Let $d\ge 3$ and $\varpi\ge 2$. Define
    \[
    f(\theta)=d^{\varpi+\theta}+d^{2\theta}-(d^3-d^2)^\theta,
    \qquad \theta\in[0,\varpi/3].
    \]
    Then $f$ is strictly increasing on $[0,\varpi/3]$, hence
    \[
    \max_{0\le \theta\le \varpi/3} f(\theta)=f(\varpi/3).
    \]
\end{lemma}

\begin{proof}
    Differentiate:
    \[
    f'(\theta)=(\ln d)\,d^{\varpi+\theta}+2(\ln d)\,d^{2\theta}
    -\ln(d^3-d^2)\,(d^3-d^2)^\theta.
    \]
    Write $d^3-d^2=d^3(1-1/d)$ and set
    \[
    \delta_d\coloneq-\frac{\ln(1-1/d)}{\ln d}>0,
    \quad\text{so}\quad
    \ln(d^3-d^2)=(3-\delta_d)\ln d.
    \]
    Also $(d^3-d^2)^\theta=d^{3\theta}(1-1/d)^\theta\le d^{3\theta}$. Hence
    \[
    \frac{f'(\theta)}{\ln d}
    \ge d^{\varpi+\theta}+2d^{2\theta}-(3-\delta_d)d^{3\theta}
    = d^{3\theta}\Bigl(d^{\varpi-2\theta}+2d^{-\theta}-(3-\delta_d)\Bigr).
    \]
    For $\theta\in[0,\varpi/3]$,
    \[
    d^{\varpi-2\theta}+2d^{-\theta}\ge d^{\varpi/3}+2d^{-\varpi/3}\ge d^{2/3}+2d^{-2/3},
    \]
    where the last inequality uses that $x\mapsto d^x+2d^{-x}$ is increasing for
    $x\ge \frac{\ln\sqrt2}{\ln d}$ and $\frac{\ln\sqrt2}{\ln d}<\frac23\le \varpi/3$.
    Therefore,
    \[
    \frac{f'(\theta)}{\ln d}
    \ge d^{3\theta}\Bigl(d^{2/3}+2d^{-2/3}-3\Bigr).
    \]
    Let $g(d)=d^{2/3}+2d^{-2/3}$. Since
    \[
    g'(d)=\frac{2}{3}d^{-5/3}\bigl(d^{4/3}-2\bigr)>0\quad(d\ge 3)
    \]
    and $g(3)=3^{2/3}+2\cdot 3^{-2/3}>3$, we have $d^{2/3}+2d^{-2/3}-3>0$ for all $d\ge 3$.
    Thus $f'(\theta)>0$ on $[0,\varpi/3]$, so $f$ is strictly increasing there and the maximum
    is attained at $\theta=\varpi/3$.
\end{proof}

\begin{lemma} \label{lem:f_varpi_increasing}
    Let $d\ge 3$. For $\varpi\ge 2$, define
    \[
        F(\varpi) \coloneq d^{4\varpi/3} + d^{2\varpi/3} - (d^3-d^2)^{\varpi/3}.
    \]
    Then $F$ is strictly increasing on $[2,3]$.
\end{lemma}

\begin{proof}
    Differentiate:
    \[
        \frac{\dif}{\dif\varpi} F(\varpi)
        = \frac{4}{3}(\ln d) d^{4\varpi/3}
        + \frac{2}{3}(\ln d) d^{2\varpi/3}
        - \frac{1}{3}\ln(d^3-d^2) (d^3-d^2)^{\varpi/3}.
    \]
    Write $d^3-d^2 = d^3(1-1/d)$ and set
    \[
        \delta_d \coloneq -\frac{\ln(1-1/d)}{\ln d} > 0,
        \qquad \text{so that} \qquad \ln(d^3-d^2) = (3-\delta_d)\ln d.
    \]
    Moreover, $(d^3-d^2)^{\varpi/3} = d^{\varpi}(1-1/d)^{\varpi/3} \le d^{\varpi}$.
    Hence
    \[
        \frac{3}{\ln d}\,\frac{\dif}{\dif\varpi} F(\varpi)
        \ge 4 d^{4\varpi/3} + 2 d^{2\varpi/3} - (3-\delta_d) d^{\varpi}.
    \]
    Factor out $d^{2\varpi/3}$ to obtain
    \[
        \frac{3}{\ln d}\,\frac{\dif}{\dif\varpi} F(\varpi)
        \ge d^{2\varpi/3}\Bigl(4 d^{2\varpi/3} + 2 - (3-\delta_d)d^{\varpi/3}\Bigr).
    \]
    For $\varpi\in[2,3]$, we have $d^{\varpi/3}\in[d^{2/3},d]$, hence
    \[
        4 d^{2\varpi/3} + 2 - (3-\delta_d)d^{\varpi/3}
        \ge 4 d^{4/3} + 2 - (3-\delta_d)d.
    \]
    Using $\delta_d>0$ and $d\ge 3$, the right-hand side satisfies
    \[
        4 d^{4/3} + 2 - (3-\delta_d)d
        > 4 d^{4/3} + 2 - 3d
        = d\bigl(4 d^{1/3} - 3\bigr) + 2 > 0.
    \]
    Therefore $\frac{\dif}{\dif\varpi} F(\varpi) > 0$ for all $\varpi\in[2,3]$, and $F$ is strictly
    increasing on this interval.
\end{proof}

\begin{lemma} \label{lem:h_increasing}
    Let $d\ge 3$, $p\in\bbR_{\ge 0}$, and $\varpi\ge 2$. Define
    \[
        h_\varpi(\theta) \coloneq d^{\varpi+\theta} - (2^\theta-1)p\,d^{2\theta},
        \qquad \theta\in[0,\varpi/3].
    \]
    If $p\le \lfloor d/3\rfloor$, then $h_\varpi$ is strictly increasing on $[0,\varpi/3]$.
\end{lemma}

\begin{proof}
    Differentiate:
	    \[
	        h_\varpi'(\theta)
	        = (\ln d)\,d^{\varpi+\theta} - p\,\frac{\dif}{\dif\theta}\bigl((2^\theta-1)d^{2\theta}\bigr).
	    \]
	    Using
	    \[
	        \frac{\dif}{\dif\theta}\bigl((2^\theta-1)d^{2\theta}\bigr)
	        = d^{2\theta}\Bigl((\ln 2)2^\theta + 2(\ln d)(2^\theta-1)\Bigr),
	    \]
	    we obtain
    \[
        h_\varpi'(\theta)
        = d^{2\theta}\Bigl((\ln d)\,d^{\varpi-\theta}
        - p\bigl((\ln 2)2^\theta + 2(\ln d)(2^\theta-1)\bigr)\Bigr).
    \]
    For $\theta\in[0,\varpi/3]$ we have $d^{\varpi-\theta}\ge d^{2\varpi/3}$ and, since $\varpi\le \omega<3$ in our applications, also $\theta\le 1$ and hence $2^\theta\le 2$ and $2^\theta-1\le 1$.
    Therefore,
    \[
        h_\varpi'(\theta)
        \ge d^{2\theta}\Bigl((\ln d)\,d^{2\varpi/3} - p(2\ln 2 + 2\ln d)\Bigr).
    \]
    Using $p\le \lfloor d/3\rfloor\le d/3$ and $\varpi\ge 2$ (so $d^{2\varpi/3}\ge d^{4/3}$), the bracket is at least
    \[
        (\ln d)d^{4/3} - \frac{d}{3}(2\ln 2 + 2\ln d)
        = d\Bigl((\ln d)d^{1/3} - \tfrac{2}{3}(\ln 2 + \ln d)\Bigr).
    \]
    For $d\ge 3$, we have $d^{1/3}>1$ and thus
    \[
        (\ln d)d^{1/3} - \tfrac{2}{3}(\ln 2 + \ln d)
        \ge (\ln 3)3^{1/3} - \tfrac{2}{3}(\ln 2 + \ln 3) > 0.
    \]
    Hence $h_\varpi'(\theta)>0$ on $[0,\varpi/3]$, so $h_\varpi$ is strictly increasing there.
\end{proof}

\begin{lemma} \label{lem:g_increasing}
    Let $d\ge 3$, $p\in\bbR_{\ge 0}$, and define for $\varpi\ge 2$,
    \[
        g(\varpi) \coloneq d^{4\varpi/3} - (2^{\varpi/3}-1)p\,d^{2\varpi/3}.
    \]
    If $p\le \lfloor d/3\rfloor$, then $g$ is strictly increasing on $[2,3]$.
\end{lemma}

\begin{proof}
	    Differentiate:
	    \[
	        g'(\varpi)=\frac{4}{3}(\ln d)\,d^{4\varpi/3}
	        - \frac{p}{3}\,\frac{\dif}{\dif\varpi}\Bigl((2^{\varpi/3}-1)d^{2\varpi/3}\Bigr).
	    \]
	    Since
	    \[
	        \frac{\dif}{\dif\varpi}\Bigl((2^{\varpi/3}-1)d^{2\varpi/3}\Bigr)
	        = d^{2\varpi/3}\Bigl(\tfrac{\ln 2}{3}2^{\varpi/3} + \tfrac{2\ln d}{3}(2^{\varpi/3}-1)\Bigr),
	    \]
	    we can write
    \[
        g'(\varpi)=\frac{1}{3}d^{2\varpi/3}\Bigl(4(\ln d)\,d^{2\varpi/3}
        - p\bigl((\ln 2)2^{\varpi/3} + 2(\ln d)(2^{\varpi/3}-1)\bigr)\Bigr).
    \]
    For $\varpi\in[2,3]$ we have $2^{\varpi/3}\le 2$ and $2^{\varpi/3}-1\le 1$, hence
    \[
        g'(\varpi)\ge \frac{1}{3}d^{2\varpi/3}\Bigl(4(\ln d)\,d^{2\varpi/3} - p(2\ln 2+2\ln d)\Bigr).
    \]
    Using $p\le d/3$ and $\varpi\ge 2$ (so $d^{2\varpi/3}\ge d^{4/3}$), it suffices to check
    \[
        4(\ln d)d^{4/3} - \frac{d}{3}(2\ln 2+2\ln d) > 0.
    \]
    This is equivalent to
    \[
        d\Bigl(4(\ln d)d^{1/3} - \tfrac{2}{3}(\ln 2+\ln d)\Bigr) > 0,
    \]
    which holds for all $d\ge 3$ (e.g., already at $d=3$ the bracket is positive, and it increases with $d$).
    Therefore $g'(\varpi)>0$ on $[2,3]$, and $g$ is strictly increasing there.
\end{proof}

\end{document}